\documentclass[journal]{IEEEtran}
\usepackage{amsmath,amsfonts}
\usepackage{amsthm}
\usepackage{algorithmic}
\usepackage[ruled,linesnumbered]{algorithm2e}
\usepackage{array}
\usepackage[caption=false,font=normalsize,labelfont=sf,textfont=sf]{subfig}
\usepackage{textcomp}
\usepackage{stfloats}
\usepackage{booktabs}
\usepackage{url}
\usepackage{verbatim}
\usepackage{graphicx}
\usepackage{cite}
\usepackage{colortbl}  %彩色表格需要加载的宏包
\usepackage{xcolor}
\usepackage{url}
\usepackage{makecell}

\usepackage{multirow}

\newtheorem{assumption}{Assumption}
\newtheorem{theorem}{Theorem}

\ifCLASSINFOpdf
\else
\fi
\begin{document}
%
% paper title
% Titles are generally capitalized except for words such as a, an, and, as,
% at, but, by, for, in, nor, of, on, or, the, to and up, which are usually
% not capitalized unless they are the first or last word of the title.
% Linebreaks \\ can be used within to get better formatting as desired.
% Do not put math or special symbols in the title.
\title{Federated Graph  Learning  for  EV Charging Demand Forecasting with Personalization\\ Against Cyberattacks}
%
%
% author names and IEEE memberships
% note positions of commas and nonbreaking spaces ( ~ ) LaTeX will not break
% a structure at a ~ so this keeps an author's name from being broken across
% two lines.
% use \thanks{} to gain access to the first footnote area
% a separate \thanks must be used for each paragraph as LaTeX2e's \thanks
% was not built to handle multiple paragraphs
%

% \author{Michael~Shell,~\IEEEmembership{Member,~IEEE,}
%         John~Doe,~\IEEEmembership{Fellow,~OSA,}
%         and~Jane~Doe,~\IEEEmembership{Life~Fellow,~IEEE}% <-this % stops a space
% \thanks{M. Shell was with the Department
% of Electrical and Computer Engineering, Georgia Institute of Technology, Atlanta,
% GA, 30332 USA e-mail: (see http://www.michaelshell.org/contact.html).}% <-this % stops a space
% \thanks{J. Doe and J. Doe are with Anonymous University.}% <-this % stops a space
% \thanks{Manuscript received April 19, 2005; revised August 26, 2015.}}

\author{Yi Li,
        Renyou Xie,
        Chaojie Li,
        Yi Wang,
        Zhaoyang Dong
        % <-this % stops a space
\thanks{
This work was jointly supported by the National Natural Science Foundation of China under Grant U24B2080 and the National Natural ScienceFoundation of China under Grant 52577107. Yi Li's work is supported by the China Scholarship Council program (Project ID:202306370280).

Y. Li is with the College of Electronic and Information Engineering, Southwest University, Email: ylic204@163.com.

R. Xie is with the School of Electrical Engineering and Telecommunications, the University of New South Wales, Sydney, 2052, E-mail: renyou.xie@unsw.edu.au.

%C. Li is with the Department of Electrical Engineering, City University of Hong Kong, E-mail: cjlee.cqu@163.com.

C. Li and Z. Y. Dong are with the Department of Electrical Engineering, City University of Hong Kong, Hong Kong.
E-mail: cjlee.cqu@163.com (C. Li); zydong@ieee.org (Z. Y. Dong).

Y. Wang is with the Department of Electrical and Electronic Engineering, the University of Hong Kong, Hong Kong, E-mail:yiwang@eee.hku.hk.
% Z. Y. Dong is with Department of Electrical Engineering, City University of Hong Kong, Hong Kong. E-mail: zydong@ieee.org.
}}
% note the % following the last \IEEEmembership and also \thanks - 
% these prevent an unwanted space from occurring between the last author name
% and the end of the author line. i.e., if you had this:
% 
% \author{....lastname \thanks{...} \thanks{...} }
%                     ^------------^------------^----Do not want these spaces!
%
% a space would be appended to the last name and could cause every name on that
% line to be shifted left slightly. This is one of those "LaTeX things". For
% instance, "\textbf{A} \textbf{B}" will typeset as "A B" not "AB". To get
% "AB" then you have to do: "\textbf{A}\textbf{B}"
% \thanks is no different in this regard, so shield the last } of each \thanks
% that ends a line with a % and do not let a space in before the next \thanks.
% Spaces after \IEEEmembership other than the last one are OK (and needed) as
% you are supposed to have spaces between the names. For what it is worth,
% this is a minor point as most people would not even notice if the said evil
% space somehow managed to creep in.

% The paper headers
\markboth{Submitted to IEEE Transactions on Intelligent Transportation Systems}%
{Shell \MakeLowercase{\textit{et al.}}: A Sample Article Using IEEEtran.cls for IEEE Journals}
% The only time the second header will appear is for the odd numbered pages
% after the title page when using the twoside option.
% 
% *** Note that you probably will NOT want to include the author's ***
% *** name in the headers of peer review papers.                   ***
% You can use \ifCLASSOPTIONpeerreview for conditional compilation here if
% you desire.

% If you want to put a publisher's ID mark on the page you can do it like
% this:
%\IEEEpubid{0000--0000/00\$00.00~\copyright~2015 IEEE}
% Remember, if you use this you must call \IEEEpubidadjcol in the second
% column for its text to clear the IEEEpubid mark.

% use for special paper notices
%\IEEEspecialpapernotice{(Invited Paper)}

% make the title area
\maketitle

% As a general rule, do not put math, special symbols or citations
% in the abstract or keywords.
\begin{abstract}
Mitigating cybersecurity risk in electric vehicle (EV) charging demand forecasting plays a crucial role in the safe operation of collective EV chargings, the stability of the power grid, and the cost-effective infrastructure expansion. However, existing methods either suffer from the data privacy issue and the susceptibility to cyberattacks or fail to consider the spatial correlation among different stations. To address these challenges, a federated graph learning approach involving multiple charging stations is proposed to collaboratively train a more generalized deep learning model for demand forecasting while capturing spatial correlations among various stations and enhancing robustness against potential attacks. Firstly, for better model performance, a spatial-temporal Graph Neural Network (GNN) model is leveraged to characterize the geographic correlation among different charging stations in a federated manner. Secondly, to ensure {\color{black} robust aggregation} and deal with the data heterogeneity in a federated setting, a message passing that utilizes a global attention mechanism to aggregate personalized models for each client is proposed. Thirdly, by concerning cyberattacks, a special credit-based function is designed to mitigate potential threats from malicious clients or unwanted attacks. Extensive experiments on \textcolor{black}{three} public EV charging datasets are conducted using various deep learning techniques and federated learning methods to demonstrate the prediction accuracy and {\color{black} robustness against cyberattack} of the proposed approach.
\end{abstract}

% Note that keywords are not normally used for peerreview papers.

% For peer review papers, you can put extra information on the cover
% page as needed:
% \ifCLASSOPTIONpeerreview
% \begin{center} \bfseries EDICS Category: 3-BBND \end{center}
% \fi
%
% For peerreview papers, this IEEEtran command inserts a page break and
% creates the second title. It will be ignored for other modes.
\IEEEpeerreviewmaketitle

\section{Introduction}\label{section1}
\IEEEPARstart{A}{s} \textcolor{black}{Intelligent Transportation Systems (ITS) evolve toward hyper-connectivity and deep electrification, the interaction between vehicular traffic and energy infrastructure is being fundamentally redefined.} The growing adoption of electric vehicles (EVs) by both individuals and businesses has raised concerns about the potential strain on the grid and charging infrastructure. %Consequently, meticulous operation becomes imperative for integrating a substantial fleet of EVs into power systems. Accurate forecasting of charging infrastructure demand plays a pivotal role in this endeavor. Additionally, 
Accurate demand forecasts are essential for sustainable operations of the electric vehicle - charging station - power grid ecosystem, which, in turn, facilitates the electrification and decarbonization of the transportation sector \cite{acharya2022false}. 

{\color{black}
\subsection{ EV Charging Demand Forecasting Challenges}

% There has been extensive research conducted on EV charging demand forecasting by collecting a large number of EV charging information to train a deep learning model for better prediction performance in a centralized way~\cite{arias2016electric, yi2022electric,ma2022multistep, 10159556}. However, this centralized training paradigm faces significant challenges:

\textcolor{black}{The proliferation of sensors and communication technologies within ITS has enabled the collection of high-fidelity, real-time mobility data. Leveraging this data-driven paradigm, extensive research has been conducted on EV charging demand forecasting by aggregating large-scale charging records to train deep learning models in a centralized way~\cite{10159556,yi2022electric, ma2022multistep}.} However, this centralized training paradigm faces significant challenges:

\textit{(i) Data privacy challenges}.
Charging data collected from individual EV stations presents dual privacy concerns. From a business perspective, station-level demand patterns represent proprietary intelligence that operators are often contractually prohibited from sharing. From an individual privacy perspective, these station-level aggregates derive from charging sessions containing personal data (timestamps, location), which could potentially reveal individual user patterns, especially at smaller stations. Transmitting such data to a central server creates privacy vulnerabilities that could violate data protection regulations and undermine consumer trust.

\textit{(ii)  Data heterogeneity challenges.} Unlike traditional fuel stations that typically adhere to uniform consumption patterns, EV charging stations exhibit substantial variations influenced by geographic location, infrastructure availability, charging speed, and user preferences. This heterogeneity complicates the development of a single forecasting model that can effectively serve all stations.

\textit{(iii) Cybersecurity challenges}.
Training a deep model to acquire information from IoT devices over a communication system embedded within an EV charging station is susceptible to cyberattacks \cite{10149139attack}, such as False Data Injection (FDI) \cite{9774855}. 
Conventional defense mechanisms are often insufficient, as high-power demand-side devices like EV charging stations are not continuously monitored by grid operators.

{\color{black} \subsection{Research Gap}

Addressing these challenges requires satisfying three fundamental requirements simultaneously: privacy preservation, spatial-temporal correlation modeling, and robustness against attacks. However, these requirements create inherent conflicts.

\textbf{The first conflict: Privacy vs. Spatial Correlation.}
Centralized training can effectively capture spatial-temporal correlations among charging stations through graph-structured representations and neighborhood aggregation \cite{yu2017spatio, 10230996, wang2023predicting}. However, this fundamentally requires sharing raw charging data, violating privacy constraints. Conversely, distributed training paradigms such as federated learning \cite{McMahan2017CommunicationEfficientLO, saputra2019energy} can preserve privacy by keeping data local, but when each station trains independently, the spatial dependencies among geographically distributed stations—such as correlated demand patterns between nearby locations or stations in similar urban contexts—cannot be exploited. Recent federated graph learning approaches \cite{meng2021cross, c2} attempt to bridge this gap but face challenges in balancing information sharing and privacy protection. \textit{How can we model spatial correlations that require cross-node information exchange while maintaining strict data privacy?}

\textbf{The second conflict: Robustness vs. Personalization.}
Defending against Byzantine attacks in federated settings typically requires detecting outliers that deviate from a consensus model \cite{yin2018byzantine, pillutla2022robust}. However, EV charging stations exhibit substantial data heterogeneity due to diverse geographical locations, user behaviors, and infrastructure types. A well-functioning personalized system \cite{9699080, c8} should produce different models for different stations, making it difficult to distinguish between legitimate heterogeneity and malicious attacks. Existing Byzantine-robust methods \cite{blanchard2017machine, NEURIPS2024_bcbdc25d} focus on consensus models, while personalized methods typically use averaging-based aggregation vulnerable to attacks \cite{lyu2022privacy}. \textit{How can we achieve Byzantine robustness when there is no single consensus to compare against?}

\textbf{The third conflict: Distributed GNN vs. Attack Detection.}
Even with distributed graph neural network training that captures spatial correlations while preserving privacy, the resulting personalized models create new vulnerabilities. In a graph-based personalized federated learning framework, each client legitimately receives different parameter updates based on both its local data distribution and its position in the spatial graph. This legitimate variability significantly complicates outlier detection—an anomalous parameter vector could indicate either a malicious attack or simply a station with unique spatial context. Traditional Byzantine detection methods designed for homogeneous parameter spaces \cite{pillutla2022robust, 10018261} become ineffective in such heterogeneous settings. \textit{How can we design aggregation mechanisms that distinguish between these two scenarios?}

Existing approaches can satisfy at most two of these three requirements. Centralized GNN methods \cite{10230996, wang2023predicting} capture spatial correlations but violate privacy; standard federated learning \cite{saputra2019energy, c5} preserves privacy but ignores spatial structure; Byzantine-robust federated methods \cite{blanchard2017machine, pillutla2022robust} provide security but sacrifice personalization or fail to model spatial dependencies. No existing framework simultaneously addresses all three requirements for EV charging demand forecasting.

\subsection{Contributions}
To address the aforementioned trilemma, this paper proposes a Personalized Federated Graph Learning (PFGL) framework with three key innovations. First, we introduce a decoupled GNN architecture that separates local temporal encoding from server-side spatial aggregation on encrypted hidden representations, enabling privacy-preserving spatial correlation modeling. Second, we design a dual-weighted aggregation mechanism that combines spatial proximity and model similarity to enable personalized models while respecting spatial dependencies. Third, we propose a credit-based adaptive weighting function that achieves Byzantine robustness without sacrificing personalization by automatically downweighting anomalous updates while preserving legitimate heterogeneity.
The main contributions of this paper are summarized as follows:
\begin{enumerate}
        \item \textbf{A privacy-preserving federated graph learning framework.} We propose a novel framework that integrates spatial-temporal graph neural networks with federated learning for EV charging demand forecasting. Through a decoupled architecture that separates local temporal encoding from server-side spatial aggregation, the framework enables stations to collaboratively capture spatial-temporal patterns and inter-station correlations while protecting data privacy. This resolves the fundamental conflict between privacy preservation and spatial correlation modeling.
    
    \item \textbf{A dual-weighted personalized aggregation mechanism.} We design a message passing aggregation method that combines spatial proximity and model similarity to encourage stations with similar data distributions and geographical contexts to collaborate more effectively. This enables personalized forecasting models tailored for each station while respecting spatial dependencies. Mathematical convergence guarantees are provided under non-IID settings.
    
    \item \textbf{A credit-based Byzantine-robust aggregation method.} We propose an adaptive credit-based mechanism that ensures robustness against malicious attacks while maintaining personalization capabilities. The mechanism allows quantitative description of trust relationships among stations, where aggregation weights are tuned adaptively based on model similarity and credit parameters. This achieves the rare combination of Byzantine robustness and personalization, addressing the conflict between security and heterogeneity adaptation.
\end{enumerate}
}

The rest of this paper is organized as follows: \textcolor{black}{Section \ref{Related work} reviews related work in federated learning, graph neural networks, and byzantine-robust methods for EV charging demand forecasting;} Section \ref{section2} formulates the problem to be studied; Section \ref{section3} provides the overall framework of the proposed personalized federated graph learning and analyzes the convergence and robustness of the method; Section \ref{section4} conducts comparative experiments and analysis on the Palo Alto dataset; \textcolor{black}{Section \ref{discussion} discusses the practical applications and limitations of this study;} Section \ref{section5} draws the conclusions. {\color{black}
A summary of basic concept and notations is provided in Table \ref{table:summary of notations}.}

\section{\textcolor{black}{Related work}}
\label{Related work}

\textcolor{black}{In this section, related challenges together with solutions
for federated learning, graph neural
networks, and byzantine-robust methods for EV charging
demand forecasting are summarized.}

\subsection{\textcolor{black}{Federated Learning Method for EV Forecasting}}

\textcolor{black}{It is widely acknowledged that Federated learning (FL) can harness the information among different clients by sharing the model rather than the data, which provides a secure and private way for different stations/operators to collaboratively train EV charging demand forecasting models \cite{McMahan2017CommunicationEfficientLO}. 
In the EV demand forecasting area, few papers have adopted the federated learning technique to improve their forecasting performance. In \cite{saputra2019energy}, vanilla FedAvg is first applied to the energy demand forecasting for EV network, while federated cluster learning is further proposed to improve the training efficiency and personalized performance. Although clustering helps to improve personalized performance, the inter-cluster feature is not learned with the framework, which deteriorates the generalized performance of the global model in each cluster.}

\textcolor{black}{Similarly, in \cite{c5}, a blockchain-based hierarchical federated learning is proposed for EV charging demand forecasting, while the clustering method is used to get a personalized model for each cluster. However, cybersecurity and the correlated relationship between and inside the clusters are not considered. To protect data privacy against different types of cyberattacks for EV energy-demanding forecasting, a secure FL is proposed in \cite{9555212}, where a double authentication process is designed to avoid the attack model from the client or server. Meanwhile, \cite{c6} proposes a local differential privacy-based approach with elastic synchronization to enhance prediction accuracy while preserving privacy against honest-but-curious servers.
Despite being effective against the cyberattack, the global model is obtained by ``averaging," which fails to address the data heterogeneity. }

\subsection{\textcolor{black}{Personalized Federated Learning for EV Forecasting}}

\textcolor{black}{The data heterogeneity challenge remains a significant obstacle in federated learning frameworks. Many \textcolor{black}{personalized federated learning (PFL)}  techniques have been proposed to address this issue, including local tuning, model interpolation, clustering, multi-task learning, and transfer learning \cite{9699080,10122655,9743558, NIPS2017_6211080f, 10646558}. Although clustering methods help improve personalized performance, as demonstrated in \cite{saputra2019energy}, they often fail to learn inter-cluster features, which deteriorates the generalized performance of the global model in each cluster \cite{10646558}. 
Alternatively, PFedRe \cite{c8} tackles data heterogeneity by filtering clients based on data distribution similarity. It measures distributional divergence between clients (e.g., via model inversion and Wasserstein distance) and selectively aggregates only those found to be similar. This filtering helps shield the global model from the negative impacts of aggregating highly heterogeneous data.
Additionally, the correlated relationships between and inside clusters are rarely considered in existing approaches like \cite{c5}. }
\begin{figure}
    \centering
    \includegraphics[width=0.8\linewidth]{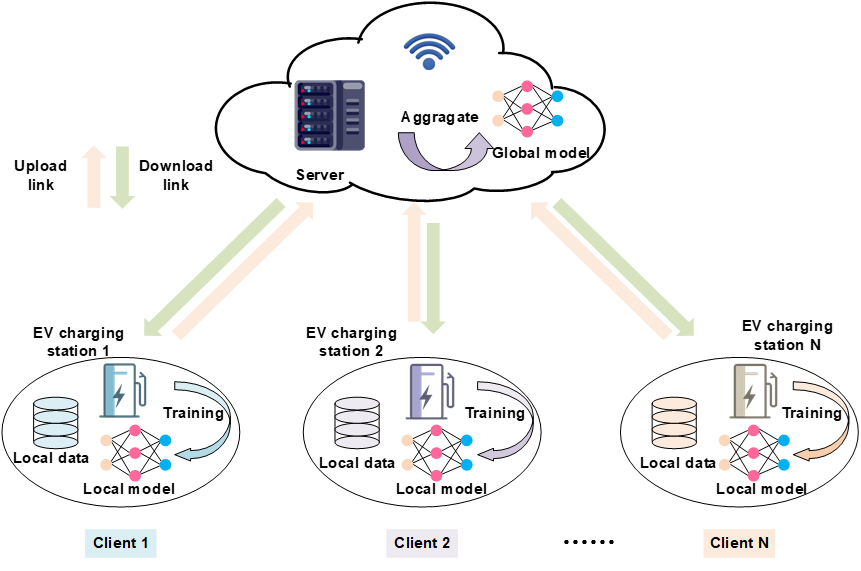}
    \caption{\textcolor{black}{The workflow of personalized federated learning}}
    \label{fig:workflow}
\end{figure}

\subsection{\textcolor{black}{Graph Neural Networks}}
\textcolor{black}{
Existing personalized federated learning methods mainly address data heterogeneity from a statistical perspective, while largely ignoring the intrinsic spatial characteristics of EV charging data. In urban environments, charging stations exhibit strong spatial and temporal correlations, which can be naturally modeled by graph neural networks (GNNs). 
 GNNs are a promising paradigm that has been applied to capture deep patterns and handle dependencies among time series across different problems in different fields, including traffic forecasting \cite{yu2017spatio, 10486198}, electricity price forecasting \cite{10693871}.
Recent studies \cite{10230996,FB2021Deep,wang2023predicting,c1,c3,c7} have integrated GNNs into EV charging demand forecasting to incorporate both spatial and temporal features, demonstrating improved generalization by leveraging data from multiple stations. However, conventional GNN training inherently relies on cross-node data interactions (e.g., neighborhood feature aggregation), which raises significant data privacy concerns due to potential exposure of sensitive node attributes or topological relationships. To mitigate these risks, federated graph learning (FGL) has emerged.
In \cite{meng2021cross}, a federated spatial-temporal model CNFGNN was proposed to effectively model complex spatial-temporal dependencies across network nodes by modeling spatial dynamics in server while temporal dynamic in edge with the aid of split learning.  It demonstrated the potential of combining GNNs with federated learning for distributed forecasting tasks. FedStar \cite{tan2023federated} is proposed to leverage the common structure information in different domains. In the EV charging demand area,
FMGCN \cite{c2} integrates spatial-temporal attention mechanisms with federated meta-learning to improve adaptability and generalization across diverse regions.} 

\textcolor{black}{Despite the federated graph learning help better personalized performance while efficient leverage the structure information, they are not immune to security vulnerabilities, requiring careful consideration of potential adversarial manipulations that could compromise forecast accuracy.
 }

\subsection{\textcolor{black}{Byzantine Robust Federated Learning Method}}

\textcolor{black}{While there are methods that address the data privacy and data heterogeneity for EV charging forecasting using the personalized FL, it still faces the risk of malicious attack on the local models when transmitting to server. 
In such cases, the commonly used ``averaging" aggregation becomes susceptible to arbitrary manipulation even if only one local model is compromised \cite{yin2018byzantine}. 
In the federated learning based EV forecasting, there is still no method that addresses this problem. In federated learning field, there are several methods that work against attacks. Typical Byzantine-robust methods classify malicious models through clustering or geometric median approaches, which can be effective against certain attacks \cite{pillutla2022robust}, \cite{10018261}. Similar robust federated learning methods including \cite{blanchard2017machine}, \cite{NEURIPS2024_bcbdc25d}. However, these methods typically train a shared model for all clients, which degrades performance when data at different stations are not independent and identically distributed (non-i.i.d). On the other hand, personalized methods address the heterogeneity issue but they incorporate ``averaging" aggregation in the server, making them vulnerable to various threat models. Few methods can achieve both personalization and robust against attacks at the same time. In \cite{lyu2022privacy}, the robustness against cyberattack and personalization is achieved simultaneously by adding a penalization term in local updating. Nevertheless, the global aggregation is still ``averaging", leading to completely separate local updating even if only one client is malicious.}

}

\vspace{-0.4cm}
\begin{table}[!h]
    {\color{black}
	\caption{Summary of main notations}
    \vspace{-0.2cm}
	%\centering
	\label{table:summary of notations}
	%\centering
    \scalebox{1}{
	\resizebox{\columnwidth}{!}{
		\begin{tabular}{l l}
			\hline \\[-3mm]
            $N$ &The total number of clients\\
            $\sigma$ & Activation function for nural networks\\
            $\alpha$ & A coefficient with range $[0,1]$ \\
            $\beta$ & Penalized coefficient for local model updating \\
            $\delta_H$ & variance of hidden representation \\
            $\rho_h$ & Lipschitz continuous with constant \\
             $t$ & Time step of forecasting model\\
			 $\tau$ & The index of $\tau$-th aggregation for FL\\
			 $\chi_i$ & A credit parameter for $i-$th client \\
             $\eta_{\tau}$ & Local learning rate.\\
			 $T$ & The number of aggregation times/ communication rounds\\
			 $K$ & The number of local update between global aggregation\\
            $V_S$ & A set of stations \\
            $V_A$ & A set of station's models \\
            $\epsilon_S$ & Edges of a graph \\
            $\epsilon_A$ & Edges of a station's model graph \\
            $\boldsymbol{\zeta}$ & Collection of the spatial weighted adjacency matrix $\{\zeta_{i,j}\}\in R^{N\times N}$  \\
            ${\boldsymbol{\xi}}$ & Collection of the similarity weighted matrix $\{{\xi}_{i,j}\}\in R^{N\times N}$  \\
            ${\boldsymbol{\lambda}}$ & Collection of the personalized weight matrix $\{{\lambda}_{i,j}\}\in R^{N\times N}$  \\
			 $\mathcal{D}_i$ & Dataset of $i$-th clients \\ 	
			 $\bold{w}$ & Global model parameter\\
			 $\bold{w}_i$ & Model parameter of $i-$th clients \\
             $\bold{w}_i^G(\tau)$ & Global model for $i-$th client at $\tau$-th aggregation.\\
              $\bold{w}_i(\tau)$ & Local model for $i-$th client at $\tau$-th aggregation.\\
            
             $f_i(\bold{w})$ & Objective function of $i$-th clients\\
             $\exp(\cdot)$ & Exponential function \\
			 $\mathcal{A}(\cdot)$ & An attention-based function \\
             $\mathcal{S}(\cdot)$ & An similarity-based function \\
			\hline
		\end{tabular}
		}
        }}
\end{table}

\section{Problem Formulation}
\label{section2}
\subsection{EV Charging Demand Forecasting}
\textcolor{black}{EV charging demand refers to the aggregated electrical energy consumption (kWh) recorded at each charging station within a given time interval, representing the total power requested by connected electric vehicles.}
The goal is to predict the future EV charging demand of $N$ charging stations in an area. We use historical demand from $N$ charging stations and calendar variables as input features. Specifically, the calendar variables refer to hour of the day, day of the week, day of the month, day of the year, month of the year, and the year. The demand observation of $i$-th charging station at time step $t$ is denoted as $x_{i}^{t}$. The corresponding calendar variables of the charging station $i$ are denoted as $z_{i}^{t}$. For convenience, we use $X_{i}^{t}$ denotes $\{x_{i}^{t}, z_{i}^{t}\}$, and $\textbf{\textit{X}}^{h}_{i}$ denotes $\{X_{i}^{t+1},\cdots,X_{i}^{t+h}\}$. The demand observation series of $p$ time steps of $i$-th station is denoted as $\textbf{\textit{Y}}^{p}_{i}=\{y_{i}^{t+h+1},\cdots,y_{i}^{t+h+p}\}$. Similarly, the demand forecast series of $p$ time steps of $i$-th station is denoted by $\hat{\textbf{\textit{Y}}}^{p}_{i}=\{\hat{y}_{i}^{t+h+1},\cdots, \hat{y}_{i}^{t+h+p}\}$. 
 Notably, the quantile loss is used for training a probabilistic forecasting model. Quantile regression models estimate the conditional quantile $q$ of the target distribution $y$, i.e., minimizing the quantile loss:
\begin{align}\label{eqn:quantile}
    L_q(y,\hat{y})= q(y-\hat{y})_+ + (1-q)(\hat{y}-y)_+,
\end{align}
where $\hat{y}$ is the predicted value, and $(y)_+=\max(0,y)$. 
%Based on (\ref{eqn:quantile}), the quantile loss penalizes the prediction error $y-\hat{y}$
%based on the quantile $q$. When $q=0.5$, the quantile loss is the same as the Mean Absolute Error. %When $q>0.5$, the loss is worse when $y>\hat{y}$, while when $q<0.5$, the loss is worse if $y<\hat{y}$.
%When using the neural network to implement quantile regression, a simple way is to fit neural networks for each quantile, which is computational costing since multiple neural networks are required to be trained. In addition, training models for each quantile are prone to result in quantile crossing, wherein the low quantile function crosses the higher one. Training a neural network model for each quantile is time-consuming and prone to result in quantile crossing, wherein the low quantile function crosses the higher one. 
In this paper, the multi-output quantile regression neural network (MQNN) is adopted to improve the computation efficiency by training a model that fits all the quantiles. Suppose $\mathcal{Q}$ is the number of quantiles, then the $\mathcal{Q}\times p$ matrix $[\hat{\textbf{\textit{Y}}}_{i}^{p(q)}]_{p,q}$ will be the output of  the MQNN of station $i$; MQNN will find the model $\bold{w}_i$ that minimizes the summed quantile loss of different quantiles, i.e.,
$\mathop{\arg\min}_{\bold{w}_i} \{\sum_{q}\sum_p L_q(y_i^p-\hat{y}_i^p)|\bold{w}_i]\}$.
 %For any probability distribution and $0 < \theta < 1$, the $\theta$-th quantile is the smallest value at which the cumulative probability mass is $\theta$. 
% The EV charging demand forecasting aims to learn a function $h(\cdot)$ that maps $p$ steps historical charging data to future $T$ steps charging demand, which can be represented as follows:
% \begin{equation}\label{eqn:graph}
%     [\mathcal{X}^{t-p+1}, \cdots,\mathcal{X}^{t};W]\xrightarrow{h(\cdot)} \hat{Y}^{T},
% \end{equation}
% \begin{equation}
%     \underset {\textbf{w}_{i}} { \operatorname {arg\,min} }\, \big\{\sum^{N}_{i=1}\|Y^{T}_{i}-\hat{Y}^{T}_{i}\|\big|\textbf{w}_{i}\big\},
% \end{equation}
% where $\textbf{w}_{i}$ denotes the model of $i$-th station, denotes the set of learnable parameters of the model, $W=\{\textbf{w}_{1},\cdots,\textbf{w}_{n}\}$. 

Additionally, to simultaneously consider the spatial-temporal correlations and enhance the generalization of our model, we go a step further by incorporating the correlation between charging station demands. This is achieved through spatial-temporal graph neural networks. Specifically, we have constructed two graphs to capture both spatial-temporal correlations and the model similarity among stations. The first spatial graph, $G_{S}=(V_{S},\epsilon_{S},\boldsymbol{\zeta})$, is designed to capture the spatial and temporal correlations, where $V_{S}$ represents the set of stations, $\epsilon_{S}$ denotes the edges, and $\boldsymbol{\zeta}=\{\zeta_{i,j}\}\in R^{N\times N}$ serves as the \textcolor{black}{spatial} weighted adjacency matrix, which is defined based on the distance. To be specifically, if the distance $d_{ij}$ between stations $i$, $j$ is below a certain threshold, $\zeta_{i,j} = 1$, which means $(V_{i}, V_{j})\in\epsilon_{S}$, otherwise $\zeta_{i,j} = 0$, $(V_{i}, V_{j})\notin\epsilon_{S}$.
The other graph is designed to capture the similarity of the forecasting model in different stations, which is denoted as the model similarity graph, $G_{A}=(V_{A},\epsilon_{A},\boldsymbol{\xi})$.  $V_{A}$ encompasses the set of station's models, $\epsilon_{A}$ signifies the edges, and $\boldsymbol{\xi}=\{\xi_{i,j}\}\in R^{N\times N}$ serves as the \textcolor{black}{similarity} weighted adjacency matrix, which is calculated based on the similarity between the models. 
In this study, we leverage both spatial connections and similarities among stations' models to encapsulate the relationships between  stations. \textcolor{black}{The proposed spatial and similarity graphs are constructed in a region-aware manner. When the studied area changes, the node set and adjacency matrices are rebuilt based on the local stations and their geographical and model-level relationships.}  %Further elaboration on the calculation of model similarity for demand modeling will be presented in Section III.
%The input of client models is $X^{t}\in R^{N\times M}$ at time $t$. 
% Consequently, the EV charging demand forecasting can be further represented as:
Consequently, the EV charging demand forecasting is modeled as:
\begin{equation}
    \{\textbf{\textit{X}}^{h}_{1}, \cdots,\textbf{\textit{X}}^{h}_{N};G_{A},G_{S}\}\xrightarrow{\mathcal{F}(\cdot)} \{\hat{\textbf{\textit{Y}}}^{p}_{1}\cdots,\hat{\textbf{\textit{Y}}}^{p}_{N}\},
\end{equation}
% \begin{equation}\label{eqn:graph}
%     \underset {\textbf{w}_{i}} { \operatorname {arg\,min} }\, \big\{\sum^{N}_{i=1}\|Y^{T}_{i,q\in Q}-\hat{Y}^{T}_{i,q\in Q}\|\big|\textbf{w}_{i},G_{A},G_{s}\big\}.
% \end{equation}
\begin{equation}
    \begin{split}\label{eqn:graph}
        \underset {\textbf{w}_{i}} { \operatorname {arg\,min} }\, &\{\sum_{q}\sum_p L_q(y_i^p-\hat{y}_i^p) | \textbf{w}_{i},G_{A},G_{s}\}.
    \end{split}
\end{equation}
\subsection{New FL for Charging Demand Forecasting}
Due to privacy concerns, the historical charging data or model $\textbf{w}_{i}$ in different stations will not be shared with other stations directly. To realize (Eq. \eqref{eqn:graph}) without sharing the embedding between stations, FL with a central server is applied. %FL enables a number of clients to train a model without sharing their data with other clients, which provides a perfect way to address the dilemma between energy operators. %Typically, federated learning consist of a central server and several clients, where the clients communicate with the server to exchange model or gradient information.  
{\color{black}Commonly used FL, such as FedAvg \cite{McMahan2017CommunicationEfficientLO}, aims to find a common model for all the clients} (refer to stations in this paper)  by solving the following optimization problem:
\begin{align}
    \min_{\bold{w}}\sum_{i\in[N]}  p_i f_i(\bold{w};\mathcal{D}_i),
\end{align}
where $f_i(\cdot)$, $\mathcal{D}_i$ are the loss function of the forecasting model and the dataset of client $i$; $p_i={|\mathcal{D}_i|}/{\sum_{i\in[N]} |\mathcal{D}_i|}$ is the weight attributed to client $i$ based on clients' dataset, and $[N]$ represents the total clients in the federated learning. 

In FL framework, the server will first send a global model to all participated clients, and each client will update their model based on the local data; after that, clients will upload their model or gradient to the server, while the server will aggregate a new global model with the collected model/gradient and use it as the initial model of the next round \cite{McMahan2017CommunicationEfficientLO}.

%It should be noticed that FedAvg can achieve a good performance where the clients with similar data distribution, whereas worse performance can be obtained than the local training when data distribution between clients is highly heterogeneity.
Due to the difference between the habits of the EV owners and the spatial distribution of the charging stations, the charging demand data in different stations vary. In such a case, we argue that an identical model for all stations may not be suitable, and personalized models by which both the common features among stations and individual features can be learned are required. A typical way to obtain a personalized model is to update the local model with local data while using the global model as a penalized term to ensure the global information is learned by the local model \cite{NEURIPS2020_f4f1f13c}, i.e.
\begin{align}\label{eqn:pfl_1}
\begin{cases}
\min_{\bold{w}_i} f_i(\bold{w}_i) + \frac{\beta_i}{2}||\bold{w}_i- \bold{w}||^2   \\
\text{s.t.} \bold{w} = \sum_{i\in [N]} p_i \bold{w}_i ,
\end{cases}
\end{align}
where $\beta_i$ is the penalization coefficient. Such a method can get an effective personalized model, while the global model is obtained using ``averaging" aggregation, which is vulnerable to malicious attacks. Once one of the clients sends the false model information to the server, the aggregated global model will be poisoned, and thus lead to a worse performance of the local model. {\color{black} In addition, the spatial correlation is not considered in Eq. \eqref{eqn:pfl_1}. Therefore, to improve the robustness of the global model to work against malicious model attacks while consider the spatial relationship among clients, a new aggregation rule is required, and Eq. \eqref{eqn:pfl_1} can be rewritten as:

%Specifically, we want to design a personalized method to help each station leverage the useful information of other stations while guaranteeing its robustness. To do so, we try to optimize the following problem:
\begin{align}\label{eqn:pfl_problem}
\begin{cases}
\min_{W} \bigl\{ \sum_{i\in[N]} \bigl(f_i(\bold{w}_i, H_i^D) + \frac{\beta_i}{2}||\bold{w}_i-\sum_{j\in [N]} \lambda_{i,j} \bold{w}_j||^2\bigl)  \bigl\} \\
\text{s.t.} \quad \lambda_{i,j}=\alpha \zeta_{i,j} + (1-\alpha) \xi_{i,j},
\end{cases}
\end{align}
 where $W=\{\bold{w}_1, \bold{w}_2,..., \bold{w}_N\}$; $H_i^D$ is hidden representation that coupled in GNN, which will be discussed in Section \ref{section3}.
 $\boldsymbol{\lambda}=\{\lambda_{i,j}\} \in R^{N \times N}$ is the personalized weight matrix that combines the spatial and similarity weighted adjacency matrix. 
And it is the final weight that used to get the global model for each client. \textcolor{black}{The abstract workflow of new FL is presented in Fig. \ref{fig:workflow}.} The key part is to get $\zeta_{i,j}$  as well as $\xi_{i,j}$, which is described in Section \ref{section3}.

\section{Proposed Personalized Federated Graph Learning Method}
\label{section3}
%This section first gives the overall framework of the proposed method and then provides implementation details of the spatial-temporal graph convolution networks, the personalized federated model with message passing, convergence analysis, and robust analysis.

% This framework is based on encoder-decoder architecture. Each charging station are denoted as a node with input time series data $X^{h}_{i}$, which are only visibel to the node $i$. The PFL framework first extracts temporal features through 1D Conv Encoder to perform node-level temporal dynamic modeling. The resulting features $H^{h}_{i}$ of each node can be regarded as the signals on the graph, $h$ is the length of the input sliding window. Then the features will upload to the central server,
% after modeling the structural and dynamics of spatio-temporal data utilizing by spatial-temporal graph convolutional networks, the spatialtemporal feature $H^{{d}_{l+1}}_{i}$ (For convenience, $H^{{d}_{l+1}}_{i}$ is denoted by $H^{d}_{i}$ in Fig.1), with initial feature $H^{h}_{i}$ of $i$th station will concatenated together $H^{D}_{i}=[H^{{d}_{l+1}}_{i},H^{h}_{i}]$ go through 1D Convs and Linear networks to aggregate personalized models for each client. 
% It should be noted that, the 1D Conv employed here is tailored for individual clients and does not involve parameter sharing. 
% In contrast, the spatial-temporal graph convolution shares parameters and employs a smaller convolution kernel to capture short-term characteristics. 
The framework of the proposed method is illustrated in Fig. \ref{fig:framework}, which involves the following steps, each marked with a corresponding stage number in the figure for clarity:
\begin{enumerate}
    \item Encoder Feature Extraction: Each EV charging station uses a 1D Conv Encoder for feature extraction. The extracted features \textcolor{black}{$H_{i}^{h}=Conv1d(\textbf{\textit{X}}^{h}_{i})$} are then sent to the central server.
    %\item Feature Upload: The resulting features $H_{i}^{h}$ are sent to a central server.
    \item Spatial-Temporal Attention Graph Convolution: The server employs spatial-temporal attention graph convolution networks to model spatial-temporal dependencies, creating the spatial-temporal feature $H^{{d}}_{i}$.% (i.e., $H^{d}_{i}$ in Fig. \ref{fig:framework}). 
    %\item Feature download: The output features of spatial-temporal graph convolution networks are downloaded by each client, and then combined with the initial feature $H^{h}_{i}$, forming $H^{D}_{i}$.
    % \item Decoder Forecasting: The output features of spatial-temporal graph convolution networks are downloaded by each client and then combined with the initial feature $H^{h}_{i}$, forming $H^{D}_{i}$ as inputs to Forecasting module which consists of 1D Convs and Linear networks.
    %\item Model Parameter Upload: The parameters of the models are uploaded to the server.
    \item Global Attention-based Message Passing: The local model parameters are uploaded to the server. The server will aggregate different global models for different clients based on the model similarities and spatial relationships among clients, which achieves message passing between different clients. After that, the server will send the aggregated global models to corresponding clients for local training.
    \item \textcolor{black}{Decoder Concatenate and Forecasting:} The output features of spatial-temporal graph convolution networks are downloaded by each client and then concat with the initial feature $H^{h}_{i}$, forming $H^{D}_{i}$ as inputs to Forecasting module which consists of 1D Convs and Linear networks.
    \item Local Training: Each client will update the local model with the received global model and local data by solving problem Eq. \eqref{eqn:local_optimization}.  % resulting in personalized models for individual EV charging stations.
\end{enumerate}
\begin{figure}
    \centering
    \includegraphics[width=9cm,height=6.8cm]
    {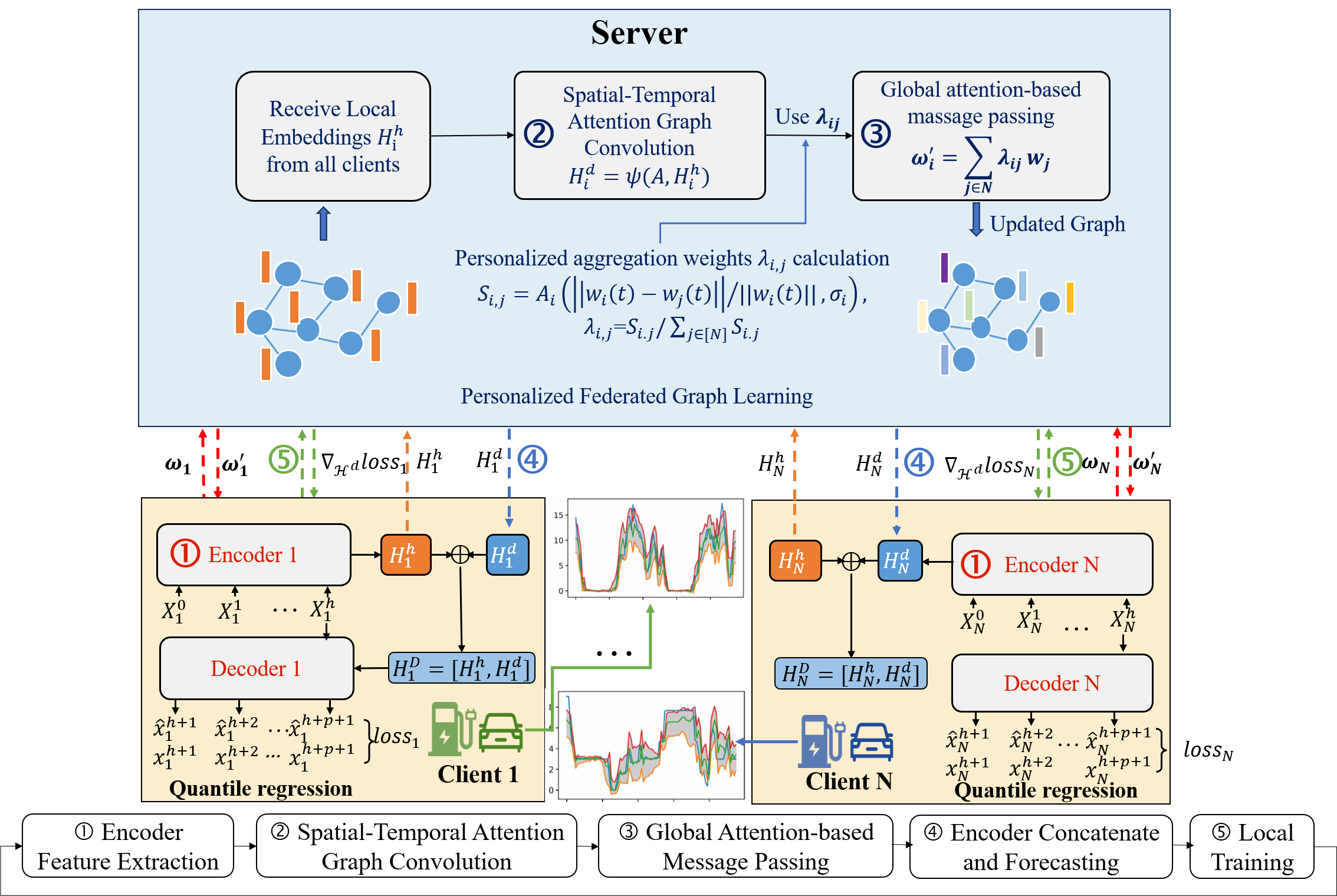}
    \setlength{\abovecaptionskip}{-0.4cm}
    \caption{\textcolor{black}{The proposed personalized federated graph learning framework of EV charging demand forecasting. First, each EV charging station model will search the model architecture locally for several epochs in the pre-search phase. Then, the architecture parameters of the model should be uploaded to the server. The models in the server can be formed as weighted graphs. Based on the similarity of spatial and model architectures, each model can acquire knowledge from other models through message passing. Then, each model can coordinate the training process and aggregate models from other distributed models for updates. Finally, the personalized models will be fine-tuned using local data to create personalized models for each individual EV charging station.}}
    \label{fig:framework}
\end{figure}

In the following, we will describe each of the above steps in detail, where 1, 2, and 3 correspond to subsection A, and 4 and 5 correspond to subsection B. Additionally, we analyzed the convergence and robustness of personalized federated learning.
\subsection{\textcolor{black}{Spatial-Temporal Attention Graph Convolution Networks (STAGCN)}}
\textcolor{black}{In this work, we propose a novel spatial-temporal attention graph convolutional network (STAGCN) to effectively model the complex spatial and temporal dependencies inherent in EV charging data. The core of our approach is a newly designed spatial-temporal attention message passing mechanism.
As illustrated in Fig. \ref{fig:framework}(a), we first apply a one-dimensional convolutional (1D Conv) layer to the input EV charging time-series data to extract high-level temporal features, denoted as $H^{h}_{i}$. These features are then passed through a STAGCN. In our framework, each client (i.e., an EV charging station) is naturally represented as a node in the graph.
The spatial module captures structural dependencies by aggregating information from neighboring clients in the graph, while the temporal module employs convolutional operations along the time axis to model temporal dynamics.} %The schematic diagram of spatial-temporal attention message passing is shown in Fig. \ref{fig:sta}.}
% \begin{figure}
%     \centering
%     \includegraphics[width=0.85\linewidth]{figures/spatial_temporal_attention.png}
%     \setlength{\abovecaptionskip}{-0.28cm}
%     \caption{\textcolor{black}{Illustration of spatial-temporal attention message passing in STAGCN. Each circle $h^t_i$ represents the embedding of client $i$ at time $t$, which collectively forms the temporal representation $\hat{H}^{(l-1)}_{i}\in R^{f_{l-1}\times d_{l-1}}$ over a sliding window. Black arrows represent spatial message passing among clients, while red dashed arrows indicate temporal attention across different time steps for the same client, capturing temporal autocorrelation.}
%  }
%     \label{fig:sta}
% \end{figure}

\textcolor{black}{We now introduce the temporal attention mechanism and detail the spatial-temporal attention graph convolution process, including the proposed spatial-temporal attention message passing strategy. This two-stage modeling framework is specifically designed to capture both temporal dependencies and spatial correlations in a unified and efficient manner, making it particularly suitable for personalized federated learning in EV charging demand forecasting.}
\textcolor{black}{\textbf{1. Temporal Attention Mechanism}}

\textcolor{black}{We first introduce a Temporal Attention Layer (TAt) to compute dynamic attention weights among different time steps. Specifically, the temporal attention can be calculated as:}

\textcolor{black}{\begin{equation}
        E = V_{e}\cdot\sigma(((\mathcal{H}^{(l)})^{T}U_{1})U_{2}(U_{3}\mathcal{H}^{(l)})+b_{e}),
    \end{equation}
    \begin{equation}
        \overline{E}_{ij} = \frac{\exp (E_{ij})}{\sum^{{d}_{l}}_{j=1}\exp{E_{ij}}},
    \end{equation}}
\textcolor{black}{where $\mathcal{H}^{(l)}=\{H^{(l)}_{1},\cdots,H^{(l)}_{N}\}\in R^{N\times f_{l}\times d_{l}}$ denotes the input of the $l+1$-th spatial-temporal attention graph neural networks. $d_{l}$ and $f_{l}$ denote the input data embedding dimension and input feature dimension of the $l$-th layer, when $l=0$, $d_{l}=h$, $H^{(l)}_{i} = H^{h}_{i}$. $V_{e}, b_{e}\in R^{d_{l}\times d_{l}}$, $U_{1}\in R^{N}$, $U_{2}\in R^{f_{l}\times N}$, $U_{3}\in R^{d_{l}}$ are learnable parameters. $E=\{E_{ij}\in R^{d_{l}\times d_{l}}\}$ represents the level of temporal dependencies between time $i$ and $j$. 
%The attention scores $\overline{E}$ capture the relative importance of each time step, allowing the model to dynamically focus on different time periods based on the temporal context. 
The attention scores $\overline{E}$ quantify the relevance of client $j$'s (history-aware) representation to client $i$'s representation at each time step, allowing the model to dynamically focus on different time periods based on the temporal context.
We then apply the normalized temporal attention matrix directly to the input and get 
\begin{equation}
    \hat{\mathcal{H}}^{(l)}= \mathcal{H}^{(l)}\cdot\overline{E}=TAt(H^{l}_{1},\cdots, H^{l}_{N}).
    \label{temporal_attention}
\end{equation}
}

\textcolor{black}{\textbf{2. Spatial-Temporal Graph Convolution}}

\textcolor{black}{To efficiently capture both spatial and temporal long-range spatial-temporal dependencies, we design a spatial-temporal attention message passing layer, guided by the temporal attention mechanism described in the previous section. This layer allows each client to selectively aggregate information from its spatial neighbors, its own temporal history, and the temporal histories of its neighbors, thus enabling fine-grained and flexible information propagation across both space and time. For spatial convolution, we adopt the GraphSAGE \cite{hamilton2017inductive}, which enables inductive representation learning by aggregating feature information from a client’s local neighborhood by using a sampling approach. In our model, the neighborhood of each client is defined by the learned personalized weighted matrix $\boldsymbol{\lambda}={\lambda_{i,j}}\in R^{N\times N}$, where each element $\lambda_{i,j}$ quantifies the  correlation strength between station $i$ and station $j$. %Based on this matrix, the set of neighbors for node $i$ is denoted as $N(i)$. Moreover, when conducting the graph convolutions, we will accompany the weighted adjacency matrix $\boldsymbol{\zeta}$ with the spatial attention matrix $S\in R^{N\times N}$ to dynamically adjust the impacting weights between stations. 
%The features of node $i$ are updated by combining the node’s own features with the features of its neighbors. 
This matrix not only defined by physical distance but also latent relationships between stations learned during training. To be specifically, when conducting the graph convolutions, we will accompany the similarity weighted adjacency matrix $\boldsymbol{\xi}$ with the spatial weighted matrix $\boldsymbol{\zeta}\in R^{N\times N}$ to dynamically adjust the impacting weights between stations.  The resulting attention matrix $\boldsymbol{\lambda}$ captures both data-driven and spatial correlations, and its detailed construction and learning process is described in the following subsection. 
%A detailed explanation of the construction and learning of the attention matrix $\boldsymbol{\lambda}$ will be provided in the next section. 
The feature representation of client $i$ is then updated by aggregating its own features with the features of its attention-weighted neighbors, enabling the model to capture dependencies more effectively.
The spatial-temporal aggregation mechanism consists of the following steps:}
% \textcolor{black}{
%     \begin{equation} \mathcal{H}^{(l)}_{N(i)} = \text{AGGREGATE}\left({\mathcal{\hat{H}}^{(l-1)}_{v} \mid v \in N(i)}, \boldsymbol{\lambda}\right), 
%     \label{aggregation_1}
%     \end{equation}
% }
% \textcolor{black}{
%     \begin{equation}
%         \mathcal{H}^{(l)}_{i} = \sigma(W^{(l)}\cdot CONCAT(\mathcal{\hat{H}}^{(l-1)}_{i},\mathcal{H}^{(l)}_{N(i)})),
%     \label{aggregation_2}
%     \end{equation}
% }
$\mathcal{H}^{(l)}_{N(i)} = \text{AGGREGATE}\left({\mathcal{\hat{H}}^{(l-1)}_{v} \mid v \in N(i)}, \boldsymbol{\lambda}\right)$, $\mathcal{H}^{(l)}_{i} = \sigma(W^{(l)}\cdot CONCAT(\mathcal{\hat{H}}^{(l-1)}_{i},\mathcal{H}^{(l)}_{N(i)}))$, 
\textcolor{black}{where $\mathcal{\hat{H}}^{(l-1)}_{i}$ denotes the feature embedding of client $i$ at the $l-1$-th layer, $\mathcal{H}^{(l)}_{N(i)}$ represents the aggregated feature embedding from its neighbors $N(i)$. In our implementation, the aggregation function is set to the element-wise mean over the neighboring clients, with the neighbors being determined by the non-zero entries in the  attention matrix $\boldsymbol{\lambda}$. }

\textcolor{black}{After the graph convolution operations,
we incorporate a standard convolution layer in the temporal
dimension. The temporal convolution operations can update
the station’s embedding by involving the information at the
neighboring time slice, represented as:}
% \textcolor{black}{
% \begin{equation}
%     \mathcal{H}^{l+1} = ReLU(\Phi*(ReLU(\text{GraphSAGE}(\mathcal{H}^{l})).
%     \label{gnn}
% \end{equation}
% }
$\mathcal{H}^{l+1} = ReLU(\Phi*(ReLU(\text{GraphSAGE}(\mathcal{H}^{l}))$.
\textcolor{black}{In the temporal convolution operation,
$*$ denotes a standard convolution, $\Phi$ is the
parameter of the temporal dimension convolution kernel,
and the activation function is ReLU. Then the final
output spatial-temporal feature $\mathcal{H}^{d}=\mathcal{H}^{l+1}$ and $\mathcal{H}^{h}$ will be concatenated together and input to the Decoder to forecast the multi-quantiles output.}

\textbf{Noted:} Most existing spatial-temporal GNNs combine GCNs with RNNs or CNNs to model spatial and temporal dependencies. However, RNN-based models suffer from limited temporal receptive fields, often missing long-range dependencies, while CNN-based methods face a trade-off between kernel size and depth, impacting efficiency and expressiveness. To overcome these limitations, we propose a spatial-temporal attention message passing layer that explicitly aggregates information from both spatial neighbors and their historical states through a temporal attention mechanism. This attention-guided aggregation captures spatial correlations and temporal dynamics jointly, followed by temporal convolution for enhanced temporal modeling. Our approach enables clients to effectively extract long-range, context-aware spatial-temporal patterns beyond traditional stacked architectures.

\subsection{Personalized Federated Graph Learning}
In the proposed federated graph learning framework, all clients solve the optimization problem Eq. \eqref{eqn:pfl_problem} in a federated manner to protect data privacy while obtaining personalized models. To enhance privacy protection and improve communication efficiency, we employ a hidden representation sharing strategy rather than directly sharing raw data or complete model parameters.  {\color{black}
Specifically, our approach involves the following key components:

\begin{enumerate}
    \item \textbf{Hidden Representation Extraction}: Each EV charging station $i$ extracts hidden representations $H^{h}_{i}=Conv1d(\textbf{\textit{X}}^{h}_{i})$ from its local time series data $\textbf{\textit{X}}^{h}_{i}$.
    
    \item \textbf{Spatial-Temporal Feature Learning}: These hidden representations are processed through spatial-temporal attention graph convolution networks to generate spatial-temporal features $H^{d}_{i}$.
    
    \item \textbf{Personalized Model Update}: Each client updates its local model by solving:
    \begin{align}\label{eqn:local_optimization}
    \min_{\bold{w}_i} f_i(\bold{w}_i; H^{D}_{i}) + \frac{\beta_i}{2}\|\bold{w}_i-\sum_{j\in [N]} \lambda_{i,j} \bold{w}_j\|^2,
    \end{align}
    where $H^{D}_{i}=[H^{{d}}_{i},H^{h}_{i}]$ represents the concatenated hidden representations and $f_i(\bold{w}_i; H^{D}_{i})$ indicates that the local loss function depends on these hidden representations.
\end{enumerate}

Upon receiving these local models, the server generates personalized global models $\bold{w}^G=\{\bold{w}^G_1,...,\bold{w}^G_N\}$ through aggregation, where $\bold{w}^G_i = \sum_{j \in [N]} \lambda_{i,j} \bold{w}_j$ for each client $i$. The personalized aggregation coefficient $\lambda_{i,j}$ is defined as $\lambda_{i,j}=\alpha \zeta_{i,j} + (1-\alpha) \xi_{i,j}$, combining the spatial weighted adjacency matrix $\zeta_{i,j}$ with the similarity weight $\xi_{i,j}$ that we focus on determining in this section.

To enhance resilience against potential malicious model attacks during the aggregation process, we compute the weights $\xi_{i,j}$ based on model similarity. The crucial component is a well-designed similarity function $\mathcal{S}(\bold{w}_i, \bold{w}_j, \chi_i)$ that evaluates the relationship between models $\bold{w}_i$ and $\bold{w}_j$ while incorporating a credit parameter $\chi_i$ representing the trust level station $i$ assigns to other stations. To establish a generalized framework for $\mathcal{S}(\bold{w}_i, \bold{w}_j, \chi_i)$, we introduce a mechanism-based function $\mathcal{A}(x, \chi_i)$ with the following key properties: $\mathcal{A}(0, \chi_i) = 1$, $\mathcal{A}(x, \chi_i) \leq \chi_i$ for $x \neq 0$, and it decreases monotonically with $x$, with values in $(0, 1]$. Based on $\mathcal{S}(\cdot)$ and $\mathcal{A}(\cdot)$, the similarity weight can be given as follow:
\begin{align}
\label{eqn:lambda_opt1}
\begin{cases}
& S_{i,j}=\mathcal{A}_i \biggl(\frac{|| \bold{w}_i-\bold{w}_j||}{|| \bold{w}_i||}, \chi_i\biggl), \\
& \xi_{i,j}=\mathcal{S}(\bold{w}_i, \bold{w}_j, \chi_i) = S_{i,j}/\sum_{j\in[N]} S_{i,j},
\end{cases}
\end{align}
where $\chi_i$ serves as a credit parameter representing the trust client $i$ places on other clients. It should be noticed that $\mathcal{S}(\bold{w}_i, \bold{w}_j, \chi_i)$ is a generalized definition, which is determined by $\mathcal{A}(\cdot)$, and there are many functions that can satisfy the requirements for $\mathcal{A}(\cdot)$. In this paper, we choose $\mathcal{A}_i(x,\chi_i)=\exp(-||\bold{x}||^2/\phi(\chi_i))$, therefore, we have:
\begin{align}
\begin{cases}
\mathcal{A}_i \biggl(\frac{|| \bold{w}_i-\bold{w}_j||}{|| \bold{w}_i||}, \chi_i\biggl) = \exp\biggl(-\frac{|| \bold{w}_i-\bold{w}_j||^2}{|| \bold{w}_i||^2 \phi(\chi_i)}\biggl),\\
\phi(\chi_i)= -\frac{||\bold{w}_i-\bold{w}_k||}{||\bold{w}_i|| \ln \chi_i}, k = \text{argmin}_{k,k\ne i} ||\bold{w}_i-\bold{w}_k||,
\end{cases}
\end{align}
 where $\phi(\chi_i)$ is chosen to ensure that $\mathcal{A}_i \biggl(\frac{|| \bold{w}_i-\bold{w}_j||}{|| \bold{w}_i||}, \chi_i\biggl) < \chi_i$. Based on this mechanism, malicious models will be identified as dissimilar to normal ones and thus assigned small weights during aggregation. This aggregation method is characterized as message passing, in which clients share information that combines knowledge from their neighboring clients' models and their own model.

The above attention-based method calculates the weight coefficient $\boldsymbol{\xi}$ between clients based on their data distribution. To further leverage the spatial relationships among different stations, the similarity value is reweighted with the spatial weighted adjacency matrix $\boldsymbol{\zeta}$, i.e., 
\begin{align}
    \boldsymbol{\lambda}_i = \alpha \boldsymbol{\xi}_i + (1-\alpha) \boldsymbol{\zeta}_{i},
    \label{channel_weight}
\end{align}
where $\alpha \in [0, 1]$ is a constant to balance $\boldsymbol{\xi}$ and $\boldsymbol{\zeta}$, which is called channel weight, and $\boldsymbol{\lambda}=[\boldsymbol{\lambda}_1,...,\boldsymbol{\lambda}_N]=\{\lambda_{i,j}\}$ is the ultimate personalized weight that used to aggregate the global model for each client.} Note that clients with $\xi_{i,j}=0$ will be excluded during the aggregation to guarantee the robustness of the global model. For better understanding, the proposed personalized federated learning algorithm is summarized in Algorithm \ref{alg:server}. Under mild assumption, the algorithm will converge to a table point, where the convergence analysis is given in Appendix A.

\begin{algorithm}
\caption{Personalized Federated Graph Learning}\label{alg:server}
\LinesNumbered
\KwIn{Initialization: $\bold{w}_i(0)=\bold{w}_0, \forall i \in [N]$, learning rate $\eta_{\tau}$, number of communication rounds $T$, number of local iterations $K$, total clients $[N]$;}
 
\For{$\tau = 0,1,\ldots,T-1$}{
    \For{$i \in [N]$}{
        {\color{black} $H^{h}_{i} \gets Conv1d(\textbf{\textit{X}}^{h}_{i})$}\;

        {\color{black} Upload $H^{h}_{i}$ to server}
    }
    
    {\color{black} // Server processes hidden representations}\\
    {\color{black} \For{$i \in [N]$}{
        $H^{d}_{i} \gets$ STGNN($\{H^{h}_{1}, H^{h}_{2}, \ldots, H^{h}_{N}\}$)\;
        Send $H^{d}_{i}$ to client $i$
    }}
    
    \For{$i \in [N]$}{
        {\color{black} // Client forms complete hidden representation}\;
        {\color{black} $H^{D}_{i} \gets [H^{d}_{i}, H^{h}_{i}]$}\;
        {\color{black} // Update model using hidden representations}
        $\bold{w}_i(\tau) \gets \texttt{Client-Update}(\bold{w}_i(\tau), H^{D}_{i})$;
    }
    
    receive $\bold{w}_i(\tau), \forall i \in [N]$ from the clients\;
    \If{all updates are received}{
        \For{$i \in [N]$}{
            calculate $\boldsymbol{\xi}_i = (\xi_{i,1}, \ldots, \xi_{i,N})$ using (\ref{eqn:lambda_opt1})\;
            $\boldsymbol{\lambda}_i \gets \alpha \boldsymbol{\xi}_i + (1 - \alpha)\boldsymbol{\zeta}_i$\;
            $\bold{w}^G_i(\tau+1) \gets \sum_{j \in [N]} \lambda_{i,j} \bold{w}_j(\tau)$\;
            broadcast $\bold{w}^G_i(\tau+1)$ to client $i$\;
        }
    }
}

Local tuning for one epoch\;

\texttt{Client-Update}($\bold{w}^G(\tau), H^{D}_{i}$)\;
receive model $\bold{w}_i^G(\tau+1)$ from server\;
\For{$k = 0$ \textbf{to} $K-1$}{
    {\color{black} // Gradient computation using hidden representations}\;
    {\color{black} $\nabla f_i(\bold{w}^k_i(\tau); H^{D}_{i}) \gets $ Compute gradient using $H^{D}_{i}$}\;
    $\bold{w}^k_i(\tau) \gets \bold{w}^k_i(\tau) - \eta_{\tau} (\nabla f_i(\bold{w}^k_i(\tau); H^{D}_{i}) + \beta_i (\bold{w}^k_i(\tau) - \bold{w}^G_i(\tau)))$\;
}
return $\bold{w}^K_i(\tau)$\;
\end{algorithm}

% {\color{black}
% \subsection{Relationship Between Hidden Representation and Similarity Coefficients}

% A crucial aspect of our approach is how hidden representations influence the model similarity calculations used for personalized aggregation. While the similarity coefficient $\xi_{i,j}$ is computed based on model parameter differences $\frac{\|\bold{w}_i-\bold{w}_j\|}{\|\bold{w}_i\|}$, these parameters are inherently shaped by the hidden representations $H^{D}_{i}$ through the local optimization process.

% Specifically, clients with similar hidden representations tend to develop similar model parameters when solving their local optimization problems. This creates an implicit relationship between hidden representation similarity and model parameter similarity:

% \begin{align}
% \text{sim}(H^{D}_{i}, H^{D}_{j}) \approx \text{sim}(\bold{w}_i, \bold{w}_j) \approx \exp\left(-\frac{\|\bold{w}_i-\bold{w}_j\|^2}{\|\bold{w}_i\|^2 \phi(\chi_i)}\right)
% \end{align}

% where $\text{sim}(H^{D}_{i}, H^{D}_{j})$ represents the similarity between hidden representations of clients $i$ and $j$.

% This relationship is particularly important for robustness against attacks, as malicious data would likely produce hidden representations that significantly differ from those of honest clients, which would then propagate to dissimilar model parameters, resulting in low similarity weights during aggregation.
% }

%**********************************************************************

\

%***************************####################*****************
\subsection{Robust Analysis Against Cyberattacks}

\textcolor{black}{In federated learning for EV charging demand forecasting, 
cyberattacks can be motivated by various adversarial objectives. 
Economically, competitors may manipulate demand forecasts to gain market 
advantages in electricity trading or charging infrastructure investment 
decisions. Operationally, attackers could aim to destabilize grid planning 
by causing overestimation (leading to unnecessary infrastructure investment) 
or underestimation (resulting in supply shortages and service disruptions).These attacks can materialize at two critical stages:
(1) during \textit{local data collection}, where compromised IoT sensors 
or communication channels can inject false charging measurements, and 
(2) during \textit{model upload}, where malicious clients can manipulate 
local model parameters before transmitting to the central server. These attacks manifest as Byzantine behavior, where 
malicious clients submit poisoned models $\bold{w}_j^{malicious}$ that 
deviate arbitrarily from honest updates. The challenge is exacerbated in 
heterogeneous environments: EV charging data naturally varies due to 
geographic location (urban vs. suburban), infrastructure type (fast vs. 
standard chargers), and user behavior (commuter vs. leisure patterns). 
This inherent heterogeneity causes legitimate clients' models to diverge 
from each other, making benign diversity indistinguishable from malicious 
manipulation using simple distance-based detection. Traditional 
Byzantine-robust aggregation methods\cite{blanchard2017machine},\cite{yin2018byzantine} thus face a 
dilemma in non-i.i.d. settings: either falsely exclude unique but honest 
clients, or tolerate poisoned models that corrupt the global aggregation.}

\textcolor{black}{Our credit-based attention mechanism addresses this by 
adapting to expected heterogeneity while detecting abnormal deviations. 
The key insight is that benign heterogeneity follows predictable spatial 
patterns (nearby stations have correlated demand), whereas attacks 
introduce spatially-incoherent, abrupt divergence.}
Specifically, the robustness of the personalized model is guaranteed by designing the attention-based function $\mathcal{A}(\cdot)$, \textcolor{black}{which calculates the similarity component $S_{i,j}$ used to determine the 
similar
% personalized aggregation 
weights $\xi_{i,j}$.} And it is chosen as
 $\mathcal{A}_i(\bold{x},\chi_i)=\exp(-||\bold{x}||^2/\phi(\chi_i))$, where $\chi_i \in (0, 1)$ is considered as a credit parameter that client $i$ put on other clients. Intuitively, large $\chi_i$ means client $i$ gives high credit to the other clients and believes that the data distribution in other clients is similar to itself, whereas small $\chi_i$ means conversely. \textcolor{black}{The core idea is that} if one of the local model $\bold{w}_j$ is poisoned, then its similarity value to model $\bold{w}_i$ will be small. Thus, the effect of the poisoned model $\bold{w}_j$  to the global model  $\bold{w}_i^G$ of client $i$ will be trivial.  \textcolor{black}{This mechanism is illustrated in Fig. \ref{fig:robustness}.}
In this paper, we assume that there is at least one client that owns a similar data distribution to client $i$, and we set the similarity value of client that is most similar to client $i$  as $\chi_i$, such that $\phi(\chi_i)$ is determined and $\xi_{i,j}$ will be determined adaptively, namely:
\begin{align}\label{eqn:sigma}
\begin{cases}
    &\mathcal{A}(\frac{||\bold{w}_i-\bold{w}_j||}{||\bold{w}_i||}, \chi_i)=\chi_i, \\
    & \text{s.t.} \quad j = \text{argmin}_{j,j\ne i} ||\bold{w}_i-\bold{w}_j||,
\end{cases}
\end{align}
and therefore $\phi(\chi_i)= -\frac{||\bold{w}_i-\bold{w}_j||}{||\bold{w}_i|| \ln \chi_i}$. %Different from the cosine similarity-based method that get the same similarity value for $\bold{w}_j$, $\bold{w}_i$, and $k\bold{w}_j$, $\bold{w}_i$
Based on (Eq. \eqref{eqn:sigma}), the normalized vector $||\bold{w}_i||$ helps to work against the scaling attack, and then 
the attention-based function $\mathcal{A}(\cdot)$ could still work even when the dimension of $\bold{w}_i$, and $\bold{w}_j$ are high, which is often the case for deep neural networks. %or the case that the similarity between the model $\bold{w}_i$ and local models cannot be distinguished due to the low dimension and small value of $\bold{w}_i$.
Based on the above analysis, we argue that the robustness of the global model for different clients can be guaranteed with the designed  $\mathcal{A}(\cdot)$. \textcolor{black}{Therefore, the proposed method is robust against common model poisoning attacks, including parameter flipping, model scaling, and Gaussian noise injection, as also validated by the false data injection experiments in Section \ref{section4}.}
% Another way to ensure the robustness of $\mathcal{A}$ is that the 
% $||\bold{w}_i({t})||$  is used to normalize the norm of difference between $\bold{w}_i$, $\bold{w}_j$, which can help $\mathcal{A}$ works even $\bold{w}_i$ with a high dimension. %In addition, we used $\frac{||\bold{w}_i-\bold{w}_j||}{||\bold{w}_i||}$ rather than $||\frac{\bold{w}_i}{||\bold{w}_i||}-\frac{\bold{w}_j}{||\bold{w}_j||}||$ to avoid the scaling attack. %that $\frac{\bold{w}_j}{||\bold{w}_j||}=\frac{k*\bold{w}_j}{||k*\bold{w}_j||}$. 

\begin{figure}
    \centering
    \includegraphics[width=0.8\linewidth]{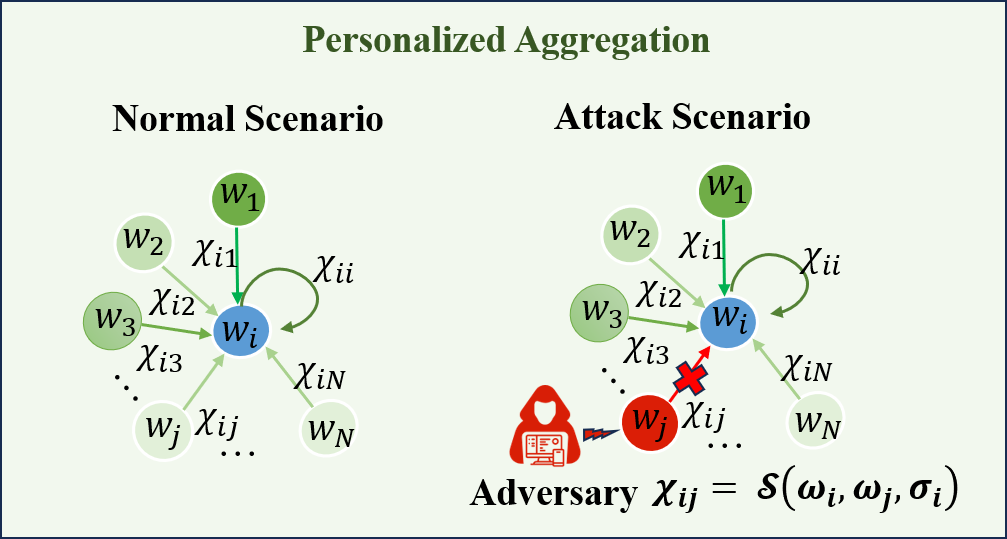}
    \vspace{-0.4cm}
    \caption{\textcolor{black}{Illustration of personalized aggregation for client i. (Left) Normal Scenario: Client $i$ aggregates parameters $w_{j}$ from benign neighbours weighted by attention $\lambda_{ij}$. (Right) Attack Secenario: An adversary poisoned client $j$'s model $w_{j}$. The attention mechanism $\xi_{ij}\sim \mathcal{A}_i(\frac{||\bold{w}_i-\bold{w}_j||}{||\bold{w}_i||}, \chi_i)$ assigns a low weight due to the large parameter difference, mitigating the attack's impact on client $i$'s update.}}
    \label{fig:robustness}
\end{figure}

\section{Case Studies}
\label{section4}
We conduct case studies on three real-world datasets—Palo Alto \cite{EVdata}, Shenzhen \cite{c1} and \textcolor{black}{UrbanEV} \cite{li2025urbanev}—capturing small- and large-scale non-IID spatial-temporal demand patterns (see Table~\ref{dataset}). To evaluate generalization, especially to unseen stations, we adopt an inductive setting on the Shenzhen dataset: 200 stations are used for training and validation, while 47 unseen stations are reserved for testing. This setup reflects real-world scenarios where new locations emerge dynamically.

We first compare our personalized federated graph learning framework with baseline methods to assess the effectiveness of GNNs and personalization. We then introduce false data injection attacks to evaluate robustness against adversarial clients and examine the model’s ability to handle heterogeneity caused by diverse but benign behaviors. We also analyze the impact of the parameter $\chi_i$.
Due to space constraints, \textcolor{black}{data processing\&distribution} and experiments involving different GNN blocks, residual structures for addressing the over-smoothing problem, sensitivity analysis of $\alpha$, and scalability evaluation with varying numbers of EV charging stations are presented in the Appendix.
\begin{table}[h]
\centering
\caption{Statistics of Real-world Datasets}
\label{dataset}
\scriptsize
\resizebox{\linewidth}{!}{
\begin{tabular}{@{}l c c c@{}}
\hline \hline
Property & Palo Alto & Shenzhen & \textcolor{black}{UrbanEV}\\
\hline
Station Num & 8 & 247 & \textcolor{black}{275}\\
Collection Period & Jan 2018--Jan 2020 & Jun 19--Jul 18, 2022 & \textcolor{black}{Sep 1, 2022--Feb 28, 2023}\\
Frequency & 30 min & 5 min & \textcolor{black}{5 min}\\
Dataset Size & 34,993 & 18,061 & \textcolor{black}{52128}\\
Features & \makecell{Time, Energy,\\ Address} & 
\makecell{Time, Energy, Address\\ Weather, Price, Occupancy, Duration} & 
\textcolor{black}{Same as Shenzhen}\\
\hline \hline
\end{tabular}
}
\end{table}
%The dataset contains various metadata on the charging transaction such as gasoline savings, charging time, and plug type. However, the metadata is excluded from the analysis because it is not available for other public datasets. We focus on using purely the electricity charging demand (kWh) for a transaction.
%  \textcolor{black}{First, we compare the personalized federated graph learning framework with some baseline to investigate the effectiveness of GNN and personalized federated learning. Second, we introduce false data injection attack in the process of clients uploading parameters to test the robustness of the provided model.
% Moreover, the robustness and the influence of $\sigma$ are investigated. The influence of CNN and increasing number of EV charging stations is also considering.}

% \textcolor{red}{In the following, we will compare our personalized federated graph learning framework with baselines to assess the effectiveness of GNN and personalized federated learning. Second, we introduce false data injection attacks during client parameter uploads to test the model's robustness. Additionally, we examine the robustness and impact of $\sigma$, as well as the effects of CNN and the increasing number of EV charging stations.}
%\textcolor{black}{Then, we compare multiple quantiles regression to single quantile regression and show the advantages of multiple quantiles regression. 
% To investigate the effectiveness of graph neural networks, we compared the federated graph neural network learning framework with the conventional federated learning framework without GNN. 
\vspace{-0.25cm}
\subsection{Baseline Approaches}
We compare our Personalized Federated Graph Learning framework (denoted as \textbf{PFGL}) with \textcolor{black}{11} benchmark methods. We include LSTM-FedAvg, LSTM-PFL, LSTM-CNN-FedAvg, and LSTM-CNN-PFL as LSTM-based federated learning baselines. \textcolor{black}{FedProx \cite{li2020federated} and pFedMe \cite{t2020personalized} are SOTA PFL baseline.} \textcolor{black}{We also consider centralized training (Central) as the Upper Bound, where a model is trained on the aggregated global dataset without federated constraints. Conversely, No\_FL trains a LSTM-based model locally on each client without any communication.}
For GNN-based methods, we include PAG-FedAvg \cite{c1} and GCRN-FedAvg \cite{c3}, which leverage spatial-temporal graphs with federated averaging. Additionally, we evaluate CNFGNN \cite{meng2021cross} and CNFGNN-PFL, which decouple temporal modeling on clients and spatial modeling on the server via a hybrid split learning and federated optimization strategy.

The detailed implementation settings, including model configurations, learning rates, are provided in the Appendix B.
We evaluate performance using the quantile loss (QS), together with two standard metrics from quantile regression: Mean Interval Length (MIL) and Interval Coverage Percentage (ICP), summarized in Table \ref{Evaluation Metrics}. Here, $\hat{y}^{t}_{i,q}$ denotes the predicted $q$-th quantile for instance $i$ at time $t$, with $q \leq q'$. \textcolor{black}{We define the quantile interval as the range between $q=0.1$ and $q=0.9$.} The testset size is denoted by $n$ and $i$ indexes test instances. Ideally, ICP should be close to 0.8, while MIL should be minimized. \textcolor{black}{QS measures the deviation between predicted and realized demand quantiles and reflects the risk of over- or under-estimating charging load. MIL captures the width of the uncertainty interval and corresponds to reserve capacity cost. ICP measures how reliably actual demand falls within the predicted range and indicates dispatch safety. In practice, lower QS reduces misforecasting risk, smaller MIL avoids excessive power reservation, and well-calibrated ICP enables safer and more economical EV charging operations, particularly under cyber-induced demand disturbances.}

% \begin{align}\label{eqn:sigma}
% QS=
% \begin{cases}
%     &\frac{1}{n}\sum^{n}_{i=1}(1-q)\|y^{t}_{i}-\hat{y}^{t}_{i,q}\|,y^{t}_{i}<\hat{y}^{t}_{i,q} \\
%     & \frac{1}{n}\sum^{n}_{i=1}q\|y^{t}_{i}-\hat{y}^{t}_{i,q}\|,y^{t}_{i}\leq\hat{y}^{t}_{i,q},
% \end{cases}
% \end{align}

% \begin{align}\label{eqn:sigma}
% ICP= &\frac{1}{n}\sum^{n}_{i=1}
% \begin{cases}
%    &1,    \hat{y}^{t}_{i,q}\leq y^{t}_{i}\leq\hat{y}^{t}_{i,q^{'}} \\
%     &0,otherwise
% \end{cases}
% \end{align}

% \begin{equation}
%     MIL = \frac{1}{n}\sum^{n}_{i=0}(\|\hat{y}^{t}_{i,q}-\hat{y}^{t}_{i,q^{'}}\|).
% \end{equation}

\begin{table}[]
\setlength{\abovecaptionskip}{0cm}
\setlength{\belowcaptionskip}{-0.6cm}
\centering
\vspace{-0.4cm}
	\caption{Evaluation Metrics}
        \label{Evaluation Metrics}
        \scalebox{1}{
	\begin{tabular}{c}
		\hline \hline
		\multirow{9}{*}{}  
QS=
\(
\begin{cases}
    \frac{1}{n}\sum^{n}_{i=1}(1-q)\|y^{t}_{i}-\hat{y}^{t}_{i,q}\|,y^{t}_{i}<\hat{y}^{t}_{i,q} \\
    \frac{1}{n}\sum^{n}_{i=1}q\|y^{t}_{i}-\hat{y}^{t}_{i,q}\|,y^{t}_{i}\leq\hat{y}^{t}_{i,q},
\end{cases}\)
\\
ICP=
\(\frac{1}{n}\sum^{n}_{i=1}
\begin{cases}
   1,    \hat{y}^{t}_{i,q}\leq y^{t}_{i}\leq\hat{y}^{t}_{i,q^{'}} \\
    0, otherwise
\end{cases}\)
\\
MIL = \(\frac{1}{n}\sum^{n}_{i=0}(\|\hat{y}^{t}_{i,q}-\hat{y}^{t}_{i,q^{'}}\|)\)\\
% MAE = $\frac{1}{n}$\sum^{n}_{i=1}\|\hat{y}^{t}_{i,q}-y^{t}_{i,q}\|\\
		\hline \hline
	\end{tabular}}
\end{table}

\subsection{Predictive Performance in Benign Environments}

The performance of the proposed method is compared with the result of the baseline. For Palo Alto dataset, the QS, MIL, and ICP of different methods are shown in Table \ref{QS}, \ref{MIL}, and \ref{ICP} respectively. The ``improved rate" represents the percentage of stations whose results are better than those of the No\_FL. \textcolor{black}{For Shenzhen and UrbanEV dataset, the metrics are shown in Table \ref{QS_2}.}

On the Palo Alto dataset, PFGL achieves the best average QS and MIL scores, outperforming baselines in 7/8 (Step-6) stations for QS, and 6/8 for MIL in Step-1  and maintaining clear advantages in long-term forecasting (Step-6). Representative 6-step-ahead predictions are illustrated in Fig. \ref{fig:LSTM_pred}, demonstrating PFGL’s superior temporal modeling capability. For large scale dataset, on the Shenzhen dataset (Table \ref{QS_2}), 
with 200 stations (47 unseen during training), PFGL improves Step-1 performance by 64.35\% (QS) and 92.42\% (MIL) over No\_FL, for Step-6, improvements are 48.21\% and 78.64\%, demonstrating strong generalization under non-IID and partial client participation settings. \textcolor{black}{On the UrbanEV dataset, characterized by higher spatial heterogeneity, PFGL prioritizes reliability, improving ICP by 76\% and QS by 57.7\% at Step-6, indicating its robustness in dense and heterogeneous urban charging networks.}

\textcolor{black}{As detailed in the Data Distribution and Heterogeneity Analysis of Appendix \ref{appendixb}, the three datasets exhibit substantially different levels of spatial heterogeneity. Palo Alto shows the highest heterogeneity ($JS = 0.4206$), followed by UrbanEV ($JS = 0.2945$), while Shenzhen is the most homogeneous ($JS = 0.2585$).
These distributional differences explain the observed behaviors in uncertainty estimation.
For the Shenzhen dataset, the relatively homogeneous demand patterns encourage PFGL to generate very sharp quantile interval with high point accuracy. However, such narrow intervals make the coverage more sensitive to local demand fluctuations, which leads to a noticeable reduction in ICP.
By contrast, both UrbanEV and Palo Alto exhibit higher heterogeneity. Under these conditions, PFGL produces wider and more conservative intervals that better capture real demand variations across stations, resulting in consistently higher ICP values close to 0.8.
This behavior demonstrates that PFGL adapts its uncertainty modeling to the underlying data distribution and remains robust under strong spatial heterogeneity.}

\textcolor{black}{To quantify the impact of federated constraints, we additionally report the performance of centralized training \cite{yi2022electric} (Upper Bound), where the model has access to all client data. PFGL remains competitive with this upper bound across datasets, indicating that the proposed personalized federated framework effectively alleviates the trade-off between data privacy and forecasting accuracy.
We further include recent personalized federated learning methods, FedProx and pFedMe, beyond the standard FedAvg. In most settings, FedProx and pFedMe consistently outperform FedAvg, confirming the necessity of personalized regularization under non-IID conditions. Overall, PFGL achieves the best QS, demonstrating the effectiveness of integrating graph modeling with personalized federated optimization.}

\textcolor{black}{In addition, while hybrid CNN/GNN architectures generally improve forecasting performance, PFGL consistently outperforms all federated baselines. Its occasional decrease in ICP indicates that the aggregation process may smooth out local variations among clients, which motivates future work on variance-aware aggregation and Bayesian GNNs to enhance uncertainty estimation without degrading point-forecast accuracy.}

\begin{table*}[]
\setlength{\abovecaptionskip}{0cm}
\centering
\color{black}{
\vspace{-0.4cm}
	\caption{QS loss for different FL with forecasting models for Palo Alto dataset}
 \vspace{-0.1cm}
        \label{QS}
        \scalebox{0.75}{
	\begin{tabular}{l l  l l c c c c c c c c c c}
		\hline \hline
		                & Method & Forecasting Model &FL Method & WEBS & HIGH & TED & MPL & CAMB & BRYANT & HAMLT & RINCO & Avg & Improved rate\\ 
		\hline
		\multirow{7}{*}{Step 1} 
            & Central & LSTM & --- & --- & --- & --- & --- & --- & ---  & --- & --- & 0.6274 & --- \\
            & No\_FL &LSTM  & --- &0.9973   & 1.1415  & 0.8761  &  0.7271 &  1.1348 & 1.0764  &0.7285  &  0.6035 & 0.9107 & 0\%    \\
            & LSTM-FedAvg & LSTM & FedAvg &0.9842 &1.0976 &0.9069 &0.7268 &1.1435 &1.0472 &0.7699 &0.6603 &0.9171 &50\%\\
            &LSTM-CNN-FedAvg & LSTM-CNN & FedAvg &0.967 &1.0862 &0.8445 &0.7131 &1.1492 &1.011 &0.701 &0.5859 &0.8823 &87.5\% \\
            &LSTM-FedProx & LSTM & FedProx &0.9916 &1.1041 &0.8417 &0.7546 &1.2224 &0.9593 &0.8029 &0.6825 &0.9199 &50\% \\
            &LSTM-pFedMe & LSTM & pFedMe &\textbf{0.919} &1.058 &0.876 &0.7653 &1.2108 &\textbf{0.9414} &0.7726 &0.683 &0.9033 &50\% \\
            &LSTM-PFL &LSTM & PFL &0.9706 &1.0907 &0.8604 &0.7196 &1.1055 &1.033 &0.7244
            &0.6117 &0.8895 &87.5\% \\
            &LSTM-CNN-PFL &LSTM-CNN & PFL &0.984 &1.0697 &0.844 &0.7129 &1.1502 &1.0159 &0.7044
            &0.5723 &0.8817 &87.5\% \\
		&CNFGNN & CNFGNN  &CNFGNN & 0.9251   &1.0634  &0.8406   &0.6934   &1.1053   &1.0228  &0.6939  &0.5736 &0.8648& \textbf{100\%}      \\
            &PAG-FedAvg &PAG & FedAvg &  0.969   &1.1096 &0.9104   &1.1082   &1.1308   &1.1674  &0.7648  &0.8964 &1.0071& \textbf{100\%}\\
            &GraphSAGE-FedAvg &GraphSAGE & FedAvg & 0.9463   &1.0682 &\textbf{0.823}   &0.6961   &\textbf{1.0883}  &1.0041  &0.6953  &0.5664 &0.8609& \textbf{100\%}\\
            &GCRN-FedAvg &GCRN & FedAvg & 0.9549   &1.0895 &0.8587   &0.7072   &1.1079  &1.0432  &0.704 &0.5711 &0.8796& \textbf{100\%}\\
            &PFGL& STAGCN & PFL &0.9215 &\textbf{1.0575} &0.8295 &\textbf{0.6899} & 1.1216 &1.0006 &\textbf{0.6892} &\textbf{0.552} &\textbf{0.8577} &\textbf{100\%}\\
		\hline
		\multirow{7}{*}{Step 6} 
            & Central & LSTM & --- & --- & --- & --- & --- & --- & ---  & --- & --- & 1.0974
 & --- \\
            & No\_FL &LSTM &---& 1.6377    & 1.8104  & 1.5343  &  1.2219 &  1.8841 & 1.8705  & 1.1031  &  1.1208 &1.5228 & 0\%    \\
            &LSTM-FedAvg &LSTM & FedAvg &1.5075 &1.6605 &1.5281 &1.214 &1.8111 &1.7175 &1.1534
            &1.2843 &1.4846 &75\%\\
            &LSTM-CNN-FedAvg &LSTM-CNN & FedAvg &1.4921 &1.6805 &1.416 &1.1844 &1.861 &1.6717 &1.0111 &1.0968 &1.4267 &\textbf{100\%}\\
            &LSTM-FedProx & LSTM & FedProx &1.3506 &1.7284 &1.3749 &1.2324 &1.7284 &1.5606 &1.1098 &1.2883 &1.4217 &62.5\% \\
            &LSTM-pFedMe & LSTM & pFedMe &1.3558 &1.7227 &1.3416 &1.1548 &1.7239 &1.619 &1.0229 &1.1726 &1.3892 &87.5\% \\
            &LSTM-PFL &LSTM &PFL &1.5107 &1.6504 &1.4995 &1.2037 &1.7949 &1.6674 &1.1186 
            &1.225 &1.4588 &75\%\\
            &LSTM-CNN-PFL &LSTM-CNN &PFL &1.4756 &1.6808 &1.4452 &1.1794 &1.839 &1.6718 &1.0113 &1.0986 &1.4252 &\textbf{100\%}\\
		&CNFGNN & CNFGNN  &CNFGNN & 1.3489   &1.5789   &1.3559   &1.1612   &1.7699   &1.5729  &0.9532  &1.0793 &1.3525 & \textbf{100\%}      \\
         &PAG-FedAvg &PAG &FedAvg & 1.4939   &1.7666 &1.5963   &1.3152   &2.004   &1.7644  &1.044  &1.2587 &1.5304& \textbf{100\%}\\
            &GraphSAGE-FedAvg &GraphSAGE &FedAvg & \textbf{1.3728}   &1.5926 &1.3267  &1.1454   &1.7651  &1.5383  &0.9799  &1.0825 &1.3504& \textbf{100\%}\\
            &GCRN-FedAvg &GCRN &FedAvg & 1.4513   &1.6698 &1.4843   &1.1903   &1.8284  &1.6623  &1.0062 &1.0826 &1.4219& \textbf{100\%}\\
		&PFGL&STAGCN & PFL &1.3768 &\textbf{1.5469} &\textbf{1.3099} &\textbf{1.1325} & \textbf{1.6801} &\textbf{1.5183} &\textbf{0.9503} &\textbf{1.029} &\textbf{1.318} &\textbf{100\%}\\		
		\hline \hline
	\end{tabular}}
    }
\end{table*}
\begin{table*}[]
\setlength{\abovecaptionskip}{0cm}
\setlength{\belowcaptionskip}{-0.6cm}
\centering
\color{black}{
\vspace{-0.4cm}
	\caption{MIL for different FL with forecasting models for Palo Alto dataset}
        \label{MIL}
        \scalebox{0.75}{
	\begin{tabular}{l l l l c c c c c c c c c c}
		\hline \hline
		                &Method &Forecasting Model   &FL Method                        & WEBS & HIGH & TED & MPL & CAMB & BRYANT & HAMLT & RINCO & Avg & Improved rate\\ 
		\hline
		\multirow{7}{*}{Step 1} 
        & Central & LSTM & --- & --- & --- & --- & --- & --- & ---  & --- & --- & 2.4308
 & --- \\
            & No\_FL &LSTM & --- & 3.2895    & 3.3887  & 2.6677  &  2.2182 &  3.4256 & 3.4951  & 2.5997  & 1.7068 & 2.8489 & 0\%    \\
            &LSTM-FedAvg &LSTM & FedAvg &3.3085 &3.5621 &2.8975 &2.1461 &3.6088 &4.1508 &\textbf{1.9084} &1.9912 &2.9467 &25\%\\
            &LSTM-CNN-FedAvg &LSTM-CNN & FedAvg &3.1142 &3.1183 &2.4359 &2.0102 &3.4061 &3.353 &2.426 &1.719 &2.6978 &87.5\%\\
            &LSTM-FedProx & LSTM & FedProx &3.1762 &3.4618 &2.845 &2.0741 &3.4641 &4.0076 &1.8764 &1.6804 &2.8232 &50\% \\
            &LSTM-pFedMe & LSTM & pFedMe &3.1545 &3.8881 &2.941 &2.941 &3.4927 &4.2935 &2.0946 &1.6871 &2.9718 &37.5\% \\
            &LSTM-PFL &LSTM & PFL &3.1057 &3.5875 &2.7625 &2.0527 &3.488 &3.7674 &2.0447 &1.7433 &2.819 &37.5\%\\
            &LSTM-CNN-PFL &LSTM-CNN &PFL &2.955 &3.2347 &2.4363 &\textbf{2.0022} &3.3019 &3.3621 &2.3433 &1.6445
            &2.66 &\textbf{100\%}\\
		& CNFGNN &CNFGNN &CNFGNN & 3.1372  &3.1707   &2.4386   &2.0509  &3.5784   &3.3498 &2.2508  &1.7596 &2.717 & 75\%  \\
        &PAG-FedAvg &PAG & FedAvg & 3.1271   &3.1644 &2.5009   &2.7687   &3.4351   &3.4793  &2.3202  &2.1262 &2.8652& \textbf{100\%}\\
            &GraphSAGE-FedAvg &GraphSAGE &FedAvg &3.1006  &3.2135 &2.439  &2.2346   &3.4751  &3.3365  &2.3352 &1.7506 &2.7356& \textbf{100\%}\\
            &GCRN-FedAvg &GCRN &FedAvg &3.1325  &3.1711 &2.53  &2.1211   &3.5486  &3.1869  &2.3589 &1.6494 &2.7123& \textbf{100\%}\\
		& PFGL &STAGCN & PFL &\textbf{2.9418} &\textbf{3.0921} &\textbf{2.2926} &2.205 & \textbf{3.2676} &\textbf{3.1491} &2.2324 &\textbf{1.4973} &\textbf{2.5847} &\textbf{100\%}\\
		\hline
		\multirow{7}{*}{Step 6} 
         & Central & LSTM & --- & --- & --- & --- & --- & --- & ---  & --- & --- & 4.2677
 & --- \\
            & No\_FL &LSTM &---& 4.8827    & 5.453  & 4.2709  &  4.1653 & 5.4922 & 6.2315  & 3.708  &  3.484& 4.711 & 0\%    \\
            &LSTM-FedAvg &LSTM &FedAvg&5.014 &5.3843 &4.4358 &3.4362 &5.3279 &6.7106 &2.9766 &3.1284
            &4.5517 &63.3\%\\
            &LSTM-CNN-FedAvg &LSTM-CNN &FedAvg  &4.2879 &4.7923 &4.1441 &3.7645 &4.8621 &5.6288 &3.1034
            &3.3035 &4.2358 &\textbf{100\%}\\
            &LSTM-FedProx & LSTM & FedProx &5.0022 &5.3836 &4.2921 &\textbf{3.3704} &5.2858 &6.5579 &\textbf{2.9011} &3.037 &4.4788 &62.5\% \\
            &LSTM-pFedMe & LSTM & pFedMe &4.2383 &5.317 &5.317 &5.317 &5.2964 &6.7345 &3.2764 &3.3657 &4.527 &87.5\% \\
            &LSTM-PFL &LSTM &PFL &4.892 &5.2963 &4.3127 &3.4221
            &5.1952 &6.239 &3.0354 &3.2 &4.4356 &63.3\%\\
            &LSTM-CNN-PFL &LSTM-CNN &PFL &4.4348 &4.7678 &3.8686 &3.6398  &4.9492 &5.3042 &3.0571 &\textbf{3.1986}
            &4.1525 &\textbf{100\%}\\
		& CNFGNN &CNFGNN &CNFGNN & 4.2927  &4.5791   &3.7782   &3.7633   &\textbf{4.5811}   &5.1661  &3.0256  &3.3389 &4.0656& \textbf{100\%}      \\
		&PAG-FedAvg &PAG & FedAvg & 4.4805   &4.7293 &3.6626   &4.3968   &4.9122   &5.3385  &3.2354  &3.2142 &4.2462& \textbf{100\%}\\
            &GraphSAGE-FedAvg &GraphSAGE &FedAvg &4.3044  &4.4738 &3.7838  &3.6303   &4.6727  &5.2961  &2.9537 &3.5716 &4.0858& \textbf{100\%}\\
            &GCRN-FedAvg &GCRN &FedAvg &4.3415  &4.7192 &4.0985  &4.0128   &4.9653  &5.1394  &3.1249 &3.6126 &4.2518& \textbf{100\%}\\	
            & PFGL & STAGCN &PFL &\textbf{4.1713}&\textbf{4.4544} &\textbf{3.701} &3.5485 &4.7736 &\textbf{5.0187}&2.9461 &3.2368 &\textbf{3.9813} &\textbf{100}\%\\
          % & PFGL &\textbf{4.1705}&4.5755 &\textbf{3.6915} &3.5863 &4.7096 &\textbf{5.0279}&3.097 &\textbf{3.1985} &\textbf{4.0072} &\textbf{100}\%\\
		\hline \hline
	\end{tabular}}
    }
\end{table*}
\begin{table*}[]
\setlength{\abovecaptionskip}{0cm}
\setlength{\belowcaptionskip}{-0.6cm}
\centering
\color{black}{
\vspace{-0.4cm}
	\caption{ICP for different FL with forecasting models for Palo Alto dataset}
        \label{ICP}
        \scalebox{0.75}{
	\begin{tabular}{l l l l c c c c c c c c c c}
		\hline \hline
		&Model                & Forecasting Model  &FL Method                         & WEBS & HIGH & TED & MPL & CAMB & BRYANT & HAMLT & RINCO & Avg & Improved rate\\ 
		\hline
		\multirow{7}{*}{Step 1} 
        & Central & LSTM & --- & --- & --- & --- & --- & --- & ---  & --- & --- & 0.8261
 & --- \\
            & No\_FL &LSTM &---& 0.8089   & 0.7804  & 0.7365  &  \textbf{0.8007} & 0.7618 & 0.8219  &0.8264  &  \textbf{0.8066} & \textbf{0.7929} & 0\%    \\
            &LSTM-FedAvg &LSTM &FedAvg &0.8378 &0.8468 &0.8448 &0.8218 &\textbf{0.8111} &0.8995 &0.464 &0.8283 &0.7942 &25\%\\
            &LSTM-CNN-FedAvg &LSTM-CNN &FedAvg &0.8312 &0.8145 &0.6424
            &0.7203 &0.7672 &0.7335 &0.8362 &0.8111 &0.7696 &25\%\\
            &LSTM-FedProx & LSTM & FedProx &0.7795 &0.8389 &0.8322 &0.4976 &0.7898 &0.7898 &0.4595 &0.821 &0.7388 &25\% \\
            &LSTM-pFedMe & LSTM & pFedMe &0.8286 &0.8141 &0.8141 &0.7987 &0.814 &0.8666 &0.7651 &0.8224 &0.8172 &37.5\% \\
            &LSTM-PFL &LSTM &PFL &0.5825 &0.8483 &0.846 &0.7287 &0.8132 &0.8807 &0.804 &0.8237 &0.7909 &37.5\%\\
            &LSTM-CNN-PFL &LSTM-CNN &PFL &0.8138 &\textbf{0.8035} &\textbf{0.8035} &0.6792 &0.7057 &0.7452 &0.8483 &0.8019 &0.779 &25\%\\
		& CNFGNN &CNFGNN &CNFGNN& \textbf{0.803}   &0.7694   &0.7515   &0.7655   &0.7846   &0.7817  &0.7886  &0.7538 &0.7755 & \textbf{62.5\%}     \\
        &PAG-FedAvg&PAG &FedAvg& 0.7614   &0.7439 &0.5873   &0.5145   &0.7161   &0.6194  &0.6109  &0.5474 &0.6376& 0\%\\
            &GraphSAGE-FedAvg &GraphSAGE &FedAvg &0.809  &0.8261&0.8328  &0.7436  &0.8132  &0.8411  &0.8089 &0.841&0.8145& 37.5\%\\
            &GCRN-FedAvg&GCRN &FedAvg &0.8062  &0.7605 &0.8073  &0.7847   &0.8158  &0.6988  &0.7618 &0.8111 &0.7808& 0\%\\
		& PFGL &STAGCN &PFL &0.7829 &0.7836 &0.7309 &0.8096 &0.7346 &\textbf{0.7821} &\textbf{0.799}&0.7876&0.7763 &37.5\%\\
		\hline
 %        \multirow{7}{*}{Step 3} 
 %            & No\_FL & 0.7689  & 0.7353  & 0.7784  &  0.7401 &  0.7749 & 0.7528 & 0.7802  &  0.7546 & 0.7606 & 0\%    \\
 %            &LSTM-FedAvg &0.6065 &0.8312 &0.8312 &0.573 &0.8063 &0.8628 &0.481 &0.7701 &0.7186 &50\%\\
 %            &LSTM-CNN-FedAvg &0.7658 &0.7416 &0.7856 &0.8115 &0.7674 &0.8183 &0.7824
 %            &0.6763 &0.7686 &50\%\\
 %            &LSTM-PFL &0.5993 &0.8223 &0.8323 &0.5388 &0.8111 &0.8748 &0.4929 &0.7563 &0.716 &50\%\\
 %            &LSTM-CNN-PFL &0.7724 &0.7432 &0.8038 &0.8129 &0.7474 &0.7676 &0.8131 &0.708 &0.771 &62.5\%\\
	% 	& CNFGNN  & \textbf{0.8019}   &0.7533   &0.7402   &0.7412   &0.7462   &\textbf{0.8075}  &\textbf{0.8035}  &0.7457 &0.7674 & 50\%      \\
	% 	& PFGL &0.8037 &\textbf{0.7797} &\textbf{0.7952} &\textbf{0.7676} &\textbf{0.7778}
 % &0.812 &0.7717 &\textbf{0.8037} &\textbf{0.7889} &\textbf{87.5\%}\\
 %            \hline
		\multirow{7}{*}{Step 6} 
        & Central & LSTM & --- & --- & --- & --- & --- & --- & ---  & --- & --- & 0.8285
 & --- \\
            & No\_FL &LSTM &--- &0.7353 &0.7713 &0.7488 &0.7734 & 0.7419&0.7902 &0.7481 &0.7592 &0.7585 &0\%    \\
            &LSTM-FedAvg &LSTM &FedAvg &0.5818 &\textbf{0.789} &0.7603 &0.4764 &\textbf{0.7955} &0.7921 &0.4762 &0.7095 &0.6726 &50\%\\
            &LSTM-CNN-FedAvg &LSTM-CNN &FedAvg &0.7194 &0.7039 &0.7642
            &0.7449 &0.7372 &0.7757 &\textbf{0.7631} &0.7555 &0.7455 &25\%\\
            &LSTM-FedProx & LSTM & FedProx &0.7748 &0.8154 &0.7662 &0.7463 &0.7776 &0.8366 &0.4702 &0.7063 &0.7063 &50\% \\
            &LSTM-pFedMe & LSTM & pFedMe &0.7408 &0.7884 &0.789 &0.735 &0.788 &0.8185 &0.476 &0.7201 &0.732 &50\% \\
            &LSTM-PFL &LSTM &PFL &0.5702 &0.8309 &0.7679 &0.4802 &0.7899 &0.8676 &0.4973 &0.7269 &0.691 &25\%\\
            &LSTM-CNN-PFL &LSTM-CNN &PFL &0.7749 &0.7466 &0.7644 &0.7616 &0.743 &\textbf{0.8038} &0.7376 &0.7111
            &0.7554 &50\%\\
		&CNFGNN & CNFGNN &CNFGNN  & 0.7404
   &0.7411  &0.7355   &0.7563   &0.7021   &0.7206  &0.764  &0.7521 &0.739 & 25\%      \\
   &PAG-FedAvg &PAG & FedAvg & 0.7748   &0.7587 &0.7344   &0.7772   &0.733   &0.7902  &0.7609  &0.698 &0.7534& 37.5\%\\
            &GraphSAGE-FedAvg &GraphSAGE &FedAvg &0.7684  &0.7364&0.7723  &0.7678  &0.741  &0.7774  &0.7604 &0.7615&0.7606& 50\%\\
            &GCRN-FedAvg &GCRN & FedAvg &0.7473  &0.7096 &0.743  &0.7653   &0.7052  &0.7563  &\textbf{0.7631} &0.7604 &0.7448& 37.5\%\\
		& PFGL &STAGCN &PFL &\textbf{0.7749} &0.7643 &\textbf{0.7896} &\textbf{0.7783} & 0.7439&0.7862 &0.7495 &\textbf{0.7607}&\textbf{0.7684} &\textbf{75\%}\\		
		\hline \hline
	\end{tabular}}
    }
\end{table*}
\begin{table*}[t]
\centering
\color{black}{
\caption{Comparison of metrics for different FL and forecasting models on Shenzhen and UrbanEV datasets}
\label{QS_2}
\scalebox{0.8}{
    \begin{tabular}{l l l l | c c c | c c c}
    \hline \hline
              & & & & \multicolumn{3}{c|}{\textbf{Shenzhen Dataset}} & \multicolumn{3}{c}{\textbf{UrbanEV Dataset}} \\
              & Method & Forecasting Method & FL Method & QS & MIL & ICP & QS & MIL & ICP \\ 
    \hline
    \multirow{12}{*}{Step 1} 
        & Central & LSTM & --- & 1.5943 & 3.214 &0.646 & 2.4116 & 6.0447 & 0.9019\\
        & No\_FL & LSTM & --- & 18.6134 & 41.4446 & 0.7106 & 2.7700 & 9.3043 & 0.8582 \\
        & LSTM-FedAvg & LSTM & FedAvg & 17.7574 & 40.4155 & 0.7117 & 3.0580 & \textbf{7.2973} & 0.8597 \\
        & LSTM-CNN-FedAvg & LSTM-CNN & FedAvg & 18.7825 & 41.9897 & 0.7334 & 3.7973 & 8.6662 & 0.8505 \\
        & LSTM-FedProx & LSTM & FedProx & 8.4304 & 7.1159 & 0.7495 & 3.0794 & 7.3984 & 0.8606 \\
        & LSTM-pFedMe & LSTM & pFedMe & 8.7668 & 6.2357 & 0.7175 & 4.0596 & 9.3959 & 0.8445 \\
        & LSTM-PFL & LSTM & PFL & 12.6632 & 21.4935 & 0.5194 & 3.0366 & 10.8336 & 0.8581 \\
        & LSTM-CNN-PFL & LSTM-CNN & PFL & 15.3383 & 45.4558 & \textbf{0.8444} & 3.0787 & 11.4880 & 0.8512 \\
        & CNFGNN & CNFGNN & CNFGNN & 7.9611 & 4.0114 & 0.6694 & 2.3151 & 24.6905 & 0.8408 \\
        & PAG-FedAvg & PAG & FedAvg & 11.2940 & 17.0661 & 0.7108 & 3.5835 & 34.6497 & 0.7594 \\
        & GraphSAGE-FedAvg & GraphSAGE & FedAvg & 7.3219 & 3.3621 & 0.6830 & 1.3156 & 14.6101 & 0.8345 \\
        & GCRN-FedAvg & GCRN & FedAvg & 7.3380 & 3.6085 & 0.6792 & 2.6515 & 28.5006 & 0.8350 \\
        & PFGL & STAGCN & PFL & \textbf{6.8046} & \textbf{3.1554} & 0.7126 & \textbf{1.3115} & 13.6568 & \textbf{0.8343} \\
    \hline
    \multirow{12}{*}{Step 6} 
        & Central & LSTM & --- & 3.9787 & 12.9453 & 0.8173 & 4.4774 & 9.0636 & 0.5633\\
        & No\_FL & LSTM & --- & 18.6698 & 39.8802 & 0.6988 & 5.4571 & \textbf{13.4520} & 0.8233 \\
        & LSTM-FedAvg & LSTM & FedAvg & 18.7478 & 41.1482 & 0.7234 & 6.1992 & 14.8294 & 0.8438 \\
        & LSTM-CNN-FedAvg & LSTM-CNN & FedAvg & 18.8383 & 40.1050 & 0.7016 & 6.7410 & 15.6433 & 0.8345 \\
        & LSTM-FedProx & LSTM & FedProx & 10.4907 & 11.9989 & 0.7583 & 6.3469 & 14.8578 & 0.8421 \\
        & LSTM-pFedMe & LSTM & pFedMe & 10.8816 & 10.7329 & 0.7257 & 7.1868 & 17.2460 & 0.8101 \\
        & LSTM-PFL & LSTM & PFL & 15.5211 & 42.9448 & \textbf{0.8138} & 6.1087 & 17.9927 & 0.7881 \\
        & LSTM-CNN-PFL & LSTM-CNN & PFL & 16.4544 & 46.1300 & 0.8206 & 6.0200 & 24.9813 & 0.8366 \\
        & CNFGNN & CNFGNN & CNFGNN & 10.0115 & 8.9972 & 0.6692 & 3.5114 & 32.0959 & 0.7922 \\
        & PAG-FedAvg & PAG & FedAvg & 13.7907 & 13.0350 & 0.4305 & 3.5761 & 34.0180 & 0.7284 \\
        & GraphSAGE-FedAvg & GraphSAGE & FedAvg & 9.6819 & 8.6313 & 0.7111 & 2.3674 & 24.6129 & 0.7925 \\
        & GCRN-FedAvg & GCRN & FedAvg & 9.9119 & 9.7647 & 0.6754 & 3.4435 & 33.8083 & 0.7887 \\
        & PFGL & STAGCN & PFL & \textbf{9.4411} & \textbf{8.5866} & 0.7002 & \textbf{2.3077} & 23.5551 & \textbf{0.7944} \\		
    \hline \hline
    \end{tabular}
} % 这个右括号闭合了 \scalebox{
}
\end{table*}
%======================================
\begin{figure*}
    \centering
    \includegraphics[width=16cm,height=5cm]{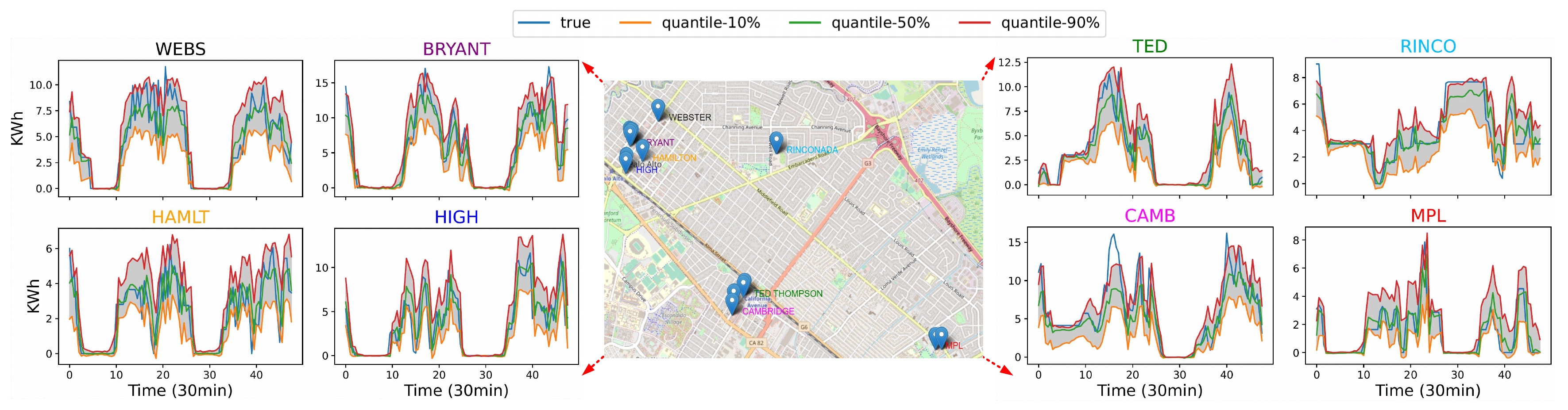}
    \setlength{\abovecaptionskip}{-0.28cm}
    \caption{Prediction result of PFGL model (6 step prediction) for stations (a) WEBS, (b) BRYANT, (c) TED, (d) RINCO, (e) HAMLT, (f) HIGH, (g) CAMB, (h) MPL.}
    \label{fig:LSTM_pred}
\end{figure*}
% \begin{table}[]
% \setlength{\abovecaptionskip}{0cm}
% \centering
% \vspace{-0.4cm}
% 	\caption{Improved rate for different Graph-based FL with forecasting models for Shenzhen dataset}
%  \vspace{-0.1cm}
%         \label{Shenzhen_improved}
%         \scalebox{0.8}{
% 	\begin{tabular}{l l c c c}
% 		\hline \hline
% 		                & Method                           & QS & MIL & ICP \\ 
% 		\hline
% 		\multirow{7}{*}{Step 1} 
%             & CNFGNN &0\%   &0\%  & 0\%\\
%             & CNFGNN-PFL &52\% &61\% &48\% \\
%             & PAG-FedAvg &4\% &1\% &56.5\% \\
%             & PAG-PFL &89\% &\textbf{83\%} &\textbf{57\%} \\
%             & GCRN-FedAvg &78\% &75.5\% &50\% \\
% 		& GCRN-PFL & 78.5\%   &73.5\%  &47.5\% \\
%             & GraphSAGE-FedAvg & 81\%   &78\% &47.5\% \\
%             & PFGL &\textbf{93\%} &79.5\% &51\% \\
% 		\hline
% 		\multirow{7}{*}{Step 6} 
%             & CNFGNN &0\%   &0\%  & 0\%\\
%             & CNFGNN-PFL &43.5\% &33\% &48\% \\
%             & PAG-FedAvg &1\% &4.5\% &15.5\% \\
%             & PAG-PFL &15.5\% &0.5\% &\textbf{56.5\%} \\
%             & GCRN-FedAvg &79\% &\textbf{68\%} &50\% \\
% 		& GCRN-PFL & 82.5\%   &66\%  &55.5\% \\
%             & GraphSAGE-FedAvg & \textbf{84\%}   &61.5\% &54.5\% \\
%             & PFGL &\textbf{84\%} &63\% &\textbf{56.5\%} \\	
% 		\hline \hline
% 	\end{tabular}}
% \end{table}
%#####################################################

\subsection{Robustness to Adversarial Clients}
\textcolor{black}{To enhance robustness against cyberattack in federated learning, }
the proposed PFL, each client will be given a personalized aggregated model by attributing different weights to the clients' model, where the weights are calculated based on the similarity between the target client and other clients.  In such a manner, clients with parameters that differ much from the target client will be attributed to less weight when performing aggregation. Naturally, when a malicious participant sends an arbitrary local model, the local model will be attributed to small or even near zero weight when it is used to aggregate the global model of another client. To testify the proposed method under false data injection attack, one of the clients is selected to send manipulated models, while three kinds of manipulating methods \cite{10552428} including parameter flipping, scaling, and Gaussian noise  are carried out to generate the threat model, as indicated in (Eq. \eqref{eqn:attack_manner}); 
\vspace{-3pt}
\begin{align}\label{eqn:attack_manner}
\bold{w}_i=
\begin{cases}
    -\bold{w}_i & \text{flipping},  \\
    scale*\bold{w}_i & \text{scaling}, \\
    \mathcal{N}(0,\epsilon\bold{I}_d) & \text{Gaussian noise}.
\end{cases}
\end{align}

\vspace{-2pt}

\begin{table*}[]
    \centering
    \color{black}{
    \caption{Average QS for different FL under different attacks (prediction step=6) for Palo Alto dataset}
    \vspace{-0.4cm}
    \label{table:attack_QS}
    \scalebox{0.75}{% % Adjust scale factor as needed for fit and readability
    \begin{tabular}{l l l c |c  c c | c c c| c c c} % l=left, c=center; @{} removes padding
        \toprule
        & & & &\multicolumn{3}{c|}{\textbf{Flipping Attack}} & \multicolumn{3}{c|}{\textbf{Scaling Attack}} & \multicolumn{3}{c}{\textbf{Gaussian Noise Attack}} \\
        %\cmidrule(lr){4-6} \cmidrule(lr){7-9} \cmidrule(lr){10-12}
        \textbf{Method}&\textbf{Forecasting Model} & \textbf{FL Method} & \textbf{No Attack} & \textbf{1 Station} & \textbf{4 Stations} & \textbf{8 Stations} & \textbf{1 Station} & \textbf{4 Stations} & \textbf{8 Stations} & \textbf{1 Station} & \textbf{4 Stations} & \textbf{8 Stations} \\
        \midrule
        No\_FL  &LSTM      & ---      & 1.4846 & ---    & ---    & ---    & ---    & ---    & ---    & ---    & ---    & ---   \\ %
        %Assuming Local corresponds to No_FL QS baseline; needs clarification if it's the LSTM-FedAvg baseline value shown. I used the LSTM-FedAvg baseline here as provided in the original 'Local' row. Adjust if No_FL has a separate baseline.
        \midrule
        LSTM-FedAvg &LSTM          & FedAvg   & 1.4846 & 2.3646 & 4.2189 & 4.4764 & 1.4867 & 1.5302 & 1.5519 & 1.4907 & 1.5302 & 1.5519 \\
        LSTM-CNN-FedAvg &LSTM-CNN      & FedAvg   & 1.4267 & 1.4262 & 1.4262 & 1.4262 & 1.4275 & 1.4275 & 1.4275 & 1.4225 & 1.4240 & 1.4263 \\
        \textcolor{black}{LSTM-FedProx} &\textcolor{black}{LSTM}          & \textcolor{black}{FedProx}  &\textcolor{black}{1.4217}   & \textcolor{black}{2.77} & \textcolor{black}{4.2071} & \textcolor{black}{4.2238} 
 & \textcolor{black}{1.6961} & \textcolor{black}{1.6939} & \textcolor{black}{1.7542} & \textcolor{black}{1.9595} & \textcolor{black}{4.002} & \textcolor{black}{4.083} \\
        LSTM-pFedme &LSTM & pFedme   & 1.3892  & 2.9572 & 4.2064 & 4.2315 & 1.6961 & 1.6939 & 1.7542 & 2.4879 & 4.041 &  4.0898 \\
        LSTM-PFL &LSTM          & PFL      & 1.4588 & 1.4351 & 1.4775 & 1.4382 & 1.4370 & 1.4268 & 1.4226 & 1.4383 & 1.4460 & 1.4352 \\
        LSTM-CNN-PFL &LSTM-CNN      & PFL      & 1.4252 & 1.4262 & 1.4262 & 1.4262 & 1.4510 & 1.4510 & 1.4510 & 1.4225 & 1.4240 & 1.4263 \\
        CNFGNN &CNFGNN        & CNFGNN   & 1.3525 & 1.3510 & 1.3562 & 1.3562 & 1.3562 & 1.3561 & \textbf{1.3378} & 1.3668 & 2.4749 & 1.3601 \\ % Assuming CNFGNN uses GNN-FL based on Table VII pattern
        PAG-FedAvg&PAG           & FedAvg   & 1.5304 & 1.3877 & 1.3877 & 1.3877 & 1.3877 & 1.3877 & 1.3877 & 1.6238 & 1.4877 & 1.5218 \\
        GraphSAGE-FedAvg&GraphSAGE     & FedAvg   & 1.3504 & 1.4343 & 1.4343 & 1.4343 & 1.3399 & 1.3426 & 1.4343 & 1.3469 & \textbf{1.3322} & 1.3386 \\
        GCRN-FedAvg &GCRN          & FedAvg   & 1.4219 & 1.4343 & 1.4343 & 1.4343 & 1.4343 & 1.4343 & 1.3415 & 1.4340 & 1.4334 & 1.4256 \\
        PFGL    &STAGCN      & PFL      & \textbf{1.318} & \textbf{1.3377} & \textbf{1.3380} & \textbf{1.3407} & \textbf{1.3411} & \textbf{1.3415} & 1.3433 & \textbf{1.3458} & 1.3397 & \textbf{1.3370} \\
        \bottomrule
    \end{tabular}%
    }
    }
\end{table*}
Since the goal of the task is to forecast the demand for EV, we use the average QS loss as the metric to evaluate its effectiveness of the proposed model to counteract malicious clients.
The attack results of the baselines and the proposed PFGL are shown in Table \ref{table:attack_QS} (Palo Alto dataset) and Table \ref{tab:attack_methods_models_shenzhen} (Shenzhen \textcolor{black}{and UrbanEV} dataset). 
% \textcolor{black}{
% On the Palo Alto dataset, the standard LSTM-FedAvg baseline shows high vulnerability to adversarial attacks. Under a flipping attack with 8 malicious clients, its QS score increases dramatically from 1.4846 (no attack) to 4.4764. Similarly, performance degrades under scaling and Gaussian noise attacks. This vulnerability is further magnified on the large-scale Shenzhen dataset, where LSTM-FedAvg fails to converge under flipping and Gaussian noise attacks involving 100 or 200 malicious stations, resulting in invalid outputs (NaN values). Even under scaling attacks, the QS score increases significantly, from 18.6698 to 45.0545, as the number of attackers grows to 200. These findings highlight the fragility of FedAvg’s naive averaging strategy in adversarial environments.} #version2
\textcolor{black}{Standard FedAvg-based models exhibit severe fragility under large-scale attacks, often failing to converge due to the direct impact of malicious updates on the global model. Notably, SOTA PFL like FedProx and pFedMe frequently underperform naive FedAvg in adversarial settings. While their proximal and regularization terms are designed to handle non-IID data.  %, they inadvertently "lock" local optimization into poisoned spaces. %Unlike FedAvg, which can partially dilute attack signals through simple averaging, 
FedProx and pFedMe are overly sensitive to local variance, they treat adversarial noise as legitimate local features, leading to catastrophic overfitting to the attack signals. }

\textcolor{black}{In contrast, PFGL consistently achieves the lowest or near-lowest QS scores across both datasets and all attack scenarios, maintaining stability even with 200 (250) malicious clients. While SOTA PFL methods inadvertently treat adversarial noise as local features, PFGL’s trust-aware design effectively distinguishes legitimate heterogeneity from malicious influence.}

The robustness of this credit-based aggregation is visualized in Fig. \ref{fig:HIGH_attack}. Whether attacks occur every 5 rounds (a) or 20 rounds (b), the algorithm effectively detects malicious submissions and resets the corresponding client weight to zero. For instance, in the 5th round, when the client 'HAMILTON' submits a poisoned model, PFGL identifies the behavior and excludes its contribution from global aggregation.

Compared to other graph-based models (e.g., CNFGNN, GraphSAGE) which offer only partial resilience, or PFL baselines (e.g., LSTM-PFL) that survive divergence but still suffer significant degradation, PFGL maintains consistent resilience. It is also important to note that some baselines appearing "unaffected" by attacks are often artifacts of early stopping or pre-attack checkpoint selection. PFGL’s ability to match or exceed local baseline performance under attack confirms that its adaptive defense is genuinely robust, neutralizing malicious influence without compromising point-forecast accuracy.
\begin{table*}[t]
\centering
\color{black}{
\caption{QS, MIL, and ICP under different attack methods and models on Shenzhen and UrbanEV datasets (Step=6)}
\label{tab:attack_methods_models_shenzhen}
\setlength{\tabcolsep}{3pt} % 略微缩小列间距以适应宽度
\scalebox{0.68}{
\begin{tabular}{clcc | ccc | ccc | ccc | ccc}
\hline \hline
& & & & \multicolumn{6}{c|}{\textbf{Shenzhen Dataset}} & \multicolumn{6}{c}{\textbf{UrbanEV Dataset}} \\ \cline{5-16}
\textbf{Attack} & \textbf{Method} & \textbf{Model} & \textbf{FL Method} & \multicolumn{3}{c|}{\textbf{Attack Station Num = 100}} & \multicolumn{3}{c|}{\textbf{Attack Station Num = 200}} & \multicolumn{3}{c|}{\textbf{Attack Station Num = 125}} & \multicolumn{3}{c}{\textbf{Attack Station Num = 250}} \\
 & & & & \textbf{QS} & \textbf{MIL} & \textbf{ICP} & \textbf{QS} & \textbf{MIL} & \textbf{ICP} & \textbf{QS} & \textbf{MIL} & \textbf{ICP} & \textbf{QS} & \textbf{MIL} & \textbf{ICP} \\ \hline
Normal & No\_FL & LSTM & Local & 18.6698 & 39.8802 & 0.6988 & 18.6698 & 39.8802 & 0.6988 & 5.4571 & 13.4520 & 0.8233 & 5.4571 & 13.4520 & 0.8233 \\ \hline
\multirow{11}{*}{Flipping} 
& LSTM-FedAvg & LSTM & FedAvg & NaN & NaN & NaN & NaN & NaN & NaN & NaN & NaN & NaN & NaN & NaN & NaN \\
& LSTM-CNN-FedAvg & LSTM-CNN & FedAvg & NaN & NaN & NaN & NaN & NaN & NaN & NaN & NaN & NaN & NaN & NaN & NaN \\
& LSTM-FedProx & LSTM & FedAvg & NaN & NaN & NaN & NaN & NaN & NaN & NaN & NaN & NaN & NaN & NaN & NaN \\
& LSTM-pFedMe & LSTM & FedAvg & NaN & NaN & NaN & NaN & NaN & NaN & NaN & NaN & NaN & NaN & NaN & NaN \\
& LSTM-PFL & LSTM & PFL & 18.3849 & 41.1156 & 0.6958 & 18.4791 & 39.2128 & 0.6899 & 6.1765 & \textbf{16.0223} & 0.8343 & 7.1656 & \textbf{17.8819} & 0.8072 \\
& LSTM-CNN-PFL & LSTM-CNN & PFL & 18.7778 & 40.1423 & 0.6920 & 18.6476 & 38.1997 & 0.6855 & 7.0703 & 19.5797 & 0.8303 & 6.7033 & 18.3378 & \textbf{0.8066} \\
& CNFGNN & CNFGNN & CNFGNN & 9.6492 & 13.0671 & 0.7080 & 9.6492 & 13.0671 & 0.7080 & 3.5513 & 29.2751 & 0.7649 & 3.4103 & 27.0044 & 0.7779 \\
& PAG-FedAvg & PAG & FedAvg & 8.2155 & 15.3904 & 0.6888 & 8.2155 & 15.3904 & 0.6888 & 3.7267 & 33.6180 & 0.7660 & 3.5982 & 31.8712 & 0.7753 \\
& GraphSAGE-FedAvg & GraphSAGE & FedAvg & 8.2155 & 15.3904 & 0.6888 & 8.2155 & 15.3904 & 0.6888 & 2.4554 & 24.6144 & \textbf{0.8035} & 2.3391 & 21.7924 & 0.8411 \\
& GCRN-FedAvg & GCRN & FedAvg & 10.1057 & 12.5223 & 0.7110 & 10.1057 & 12.5223 & \textbf{0.7110} & 2.5171 & 21.9923 & 0.8245 & 2.4908 & 21.7881 & 0.8319 \\
& PFGL-PFL & STAGCN & PFL & \textbf{7.2311} & \textbf{10.8379} & \textbf{0.7177} & \textbf{7.1906} & \textbf{10.2700} & 0.7040 & \textbf{2.3227} & 22.1071 & 0.8321 & \textbf{2.3293} & 21.7721 & 0.8390 \\ \hline
\multirow{11}{*}{Gaussian} 
& LSTM-FedAvg & LSTM & FedAvg & NaN & NaN & NaN & NaN & NaN & NaN & NaN & NaN & NaN & NaN & NaN & NaN \\
& LSTM-CNN-FedAvg & LSTM-CNN & FedAvg & NaN & NaN & NaN & NaN & NaN & NaN & NaN & NaN & NaN & NaN & NaN & NaN \\
& LSTM-FedProx & LSTM & FedProx & NaN & NaN & NaN & NaN & NaN & NaN & NaN & NaN & NaN & NaN & NaN & NaN \\
& LSTM-pFedMe & LSTM & FedProx & NaN & NaN & NaN & NaN & NaN & NaN & NaN & NaN & NaN & NaN & NaN & NaN \\
& LSTM-PFL & LSTM & PFL & 18.3849 & 41.1156 & 0.6958 & 16.9948 & 39.1417 & \textbf{0.7404} & 6.7270 & 14.6088 & 0.7745 & 6.9631 & \textbf{18.0858} & 0.8474 \\
& LSTM-CNN-PFL & LSTM-CNN & PFL & 18.7778 & 40.1423 & 0.6920 & 17.5545 & 41.4111 & \textbf{0.7404} & 7.0686 & 17.7414 & 0.8289 & 8.0409 & 20.8103 & 0.7569 \\
& CNFGNN & CNFGNN & CNFGNN & 9.9756 & 8.7336 & 0.6922 & 9.9532 & 9.1456 & 0.7042 & 3.4065 & 26.9842 & 0.7641 & 3.4545 & 27.0185 & 0.7600 \\
& PAG-FedAvg & PAG & FedAvg & 13.7947 & 10.5661 & 0.2634 & 14.0114 & 10.8334 & 0.2847 & 3.6003 & 31.4088 & 0.7759 & 3.6211 & 31.9701 & 0.7603 \\
& GraphSAGE-FedAvg & GraphSAGE & FedAvg & 9.6007 & 10.4636 & 0.7012 & 9.5381 & 11.6948 & 0.7084 & 2.5053 & 24.8911 & 0.8535 & 3.4065 & 26.9842 & 0.7441 \\
& GCRN-FedAvg & GCRN & FedAvg & 10.6089 & 8.9200 & 0.6950 & 10.3309 & 10.1570 & 0.6713 & 2.5171 & 21.9923 & 0.8245 & 2.4832 & 22.1631 & 0.8455 \\
& PFGL & STAGCN & PFL & \textbf{9.5335} & \textbf{8.6638} & \textbf{0.7121} & \textbf{8.4219} & \textbf{8.7488} & 0.7063 & \textbf{2.3061} & 21.8108 & \textbf{0.8209} & \textbf{2.3275} & 21.6271 & \textbf{0.8396} \\ \hline
\multirow{11}{*}{Scaling} 
& LSTM-FedAvg & LSTM & FedAvg & 15.0425 & 44.8933 & 0.8761 & 45.9425 & 88.1984 & 0.5929 & 6.4570 & 14.9549 & 0.8612 & 6.7472 & 14.3922 & 0.8424 \\
& LSTM-CNN-FedAvg & LSTM-CNN & FedAvg & 16.3067 & 52.0474 & 0.8751 & 45.0545 & 88.2608 & 0.6080 & 6.6924 & 15.7140 & 0.8590 & 7.0149 & 16.1381 & 0.8425 \\
& LSTM-FedProx & LSTM & FedProx & 18.7321 & 50.7004 & 0.8299 & 18.9580 & 56.7669 & 0.8772 & 6.5173 & 15.1634 & 0.8615 & 6.9076 & 15.2550 & 0.8471 \\
& LSTM-pFedMe & LSTM & FedProx & 11.9964 & 13.0272 & 0.7169 & 13.4367 & 14.0983 & 0.6283 & 7.2161 & 20.5345 & 0.8910 & 7.7030 & 22.4613 & 0.9069 \\
& LSTM-PFL & LSTM & PFL & 15.9723 & 42.6293 & 0.7739 & 18.8130 & 44.0689 & 0.7078 & 6.5139 & 16.9236 & 0.7581 & 6.5514 & 17.7666 & 0.8377 \\
& LSTM-CNN-PFL & LSTM-CNN & PFL & 16.7819 & 44.7879 & \textbf{0.7750} & 18.6233 & 42.1417 & 0.7114 & 6.7555 & 17.6447 & 0.8461 & 7.9353 & 20.3990 & 0.7605 \\
& CNFGNN & CNFGNN & CNFGNN & 9.6492 & 13.0671 & 0.7080 & 9.6492 & 13.0671 & 0.7080 & 3.4103 & 27.0044 & 0.7679 & 3.4545 & 27.0185 & 0.7600 \\
& PAG-FedAvg & PAG & FedAvg & 8.2155 & 15.3904 & 0.6888 & 8.2155 & 15.3904 & 0.6888 & 3.5982 & 31.8712 & 0.7623 & 3.7279 & 33.6345 & 0.7618 \\
& GraphSAGE-FedAvg & GraphSAGE & FedAvg & 8.2155 & 15.3904 & 0.6888 & 8.2155 & 15.3904 & 0.6888 & 2.3314 & 21.7452 & 0.8382 & 2.3354 & 21.7724 & 0.8479 \\
& GCRN-FedAvg & GCRN & FedAvg & 10.1057 & 12.5223 & 0.7010 & 10.1057 & 12.5223 & 0.7010 & 2.4908 & 21.7880 & 0.8419 & 2.4908 & 21.7880 & 0.8419 \\
& PFGL & STAGCN & PFL & \textbf{6.9592} & \textbf{10.4090} & 0.7085 & \textbf{6.7242} & \textbf{10.1656} & \textbf{0.7116} & \textbf{2.3175} & \textbf{21.6444} & 0.8376 & \textbf{2.3382} & \textbf{21.6945} & \textbf{0.8412} \\ \hline \hline
\end{tabular}
}
}
\end{table*}

\begin{figure}[h]
    \centering
    \includegraphics[width=9cm,height=6cm]
    {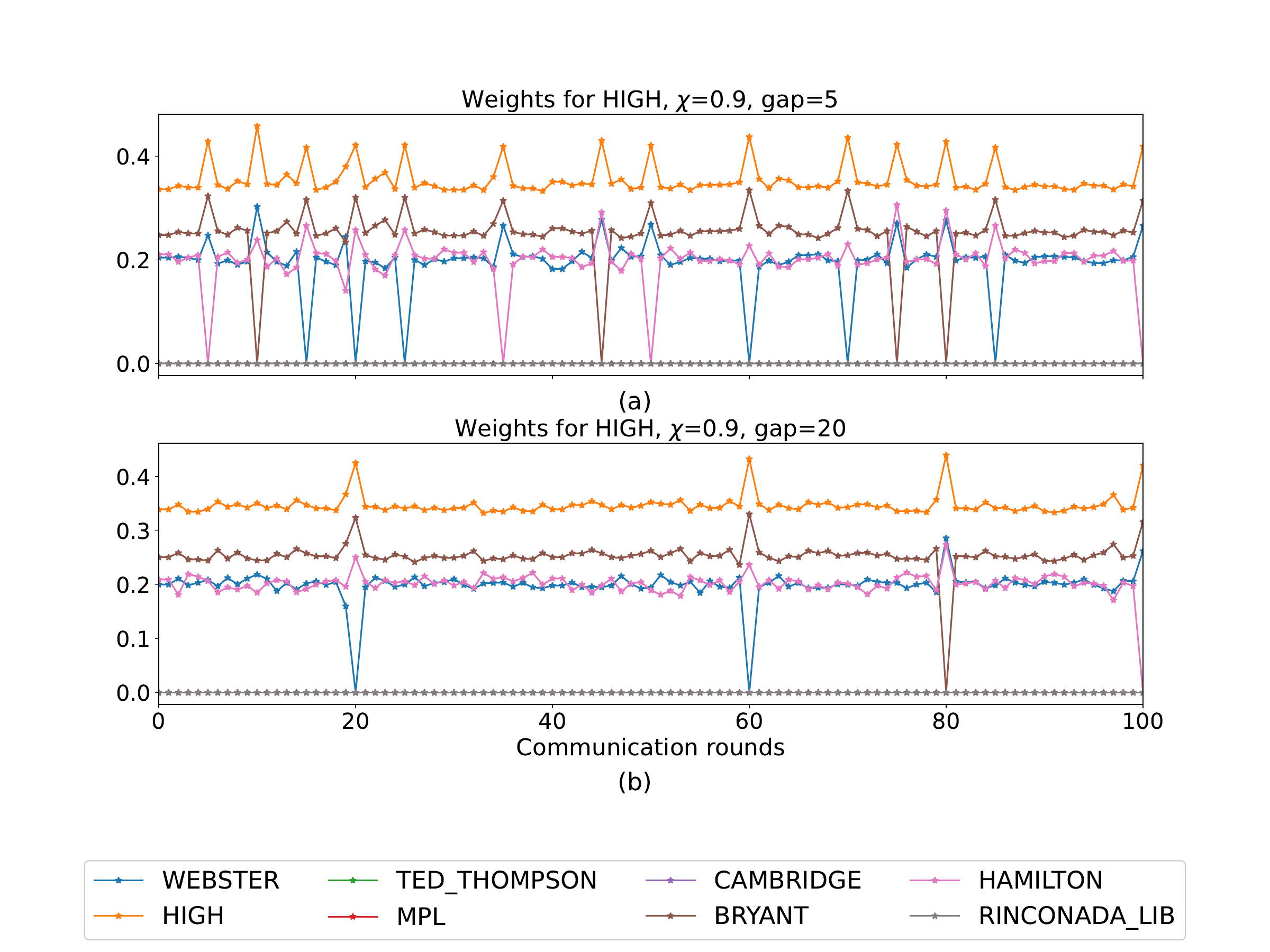}
    \setlength{\abovecaptionskip}{-0.4cm}
    \caption{Weight $\xi_{i,j}$ attributed to 'HIGH' station (a) attack happens every 5 rounds (b) attack happens every 20 rounds}
    \label{fig:HIGH_attack}
\end{figure}
\subsection{Adaptability to Benign Client Heterogeneity}
Beyond defending against explicitly malicious clients, federated systems must also remain robust to benign clients with unique, outlier behaviors due to regional or operational differences. To further validate PFGL's resilience in such heterogeneous settings, we conducted a case study using the Palo Alto dataset, focusing on clients exhibiting distinctive but non-malicious demand patterns.

As shown in Fig. \ref{fig:LSTM_pred} and Fig. \ref{fig:demand_distribution}, the RINCO station exhibits a noticeably different demand distribution compared to other stations. Despite this divergence, the RINCO station was not subjected to any attacks. 
Fig. \ref{fig:heatmap1} presents the pairwise model similarity matrix $\boldsymbol{\lambda}$ obtained by PFGL. Although RINCO shows overall lower similarity to most clients, it maintains relatively higher similarity with BRYANT and HAMILTON. This partial alignment, combined with our aggregation strategy, ensures that RINCO is correctly recognized as benign despite its distinct characteristics.  Importantly, in our method, model similarity constitutes only one component of the aggregation and anomaly assessment. As detailed in Section III (Eq. \eqref{channel_weight}), PFGL integrates both model-based similarity and spatial-based similarity, modulated via the channel weight $\alpha$. This dual-factor mechanism reduces over-reliance on any single indicator and inherently mitigates the risk of falsely penalizing benign but unique clients.

Overall, these results provide strong empirical evidence that PFGL robustly tolerates natural heterogeneity among clients and maintains accurate detection  under non-i.i.d. conditions.

\begin{figure*}
    \centering
    \includegraphics[width=16cm,height=5cm]{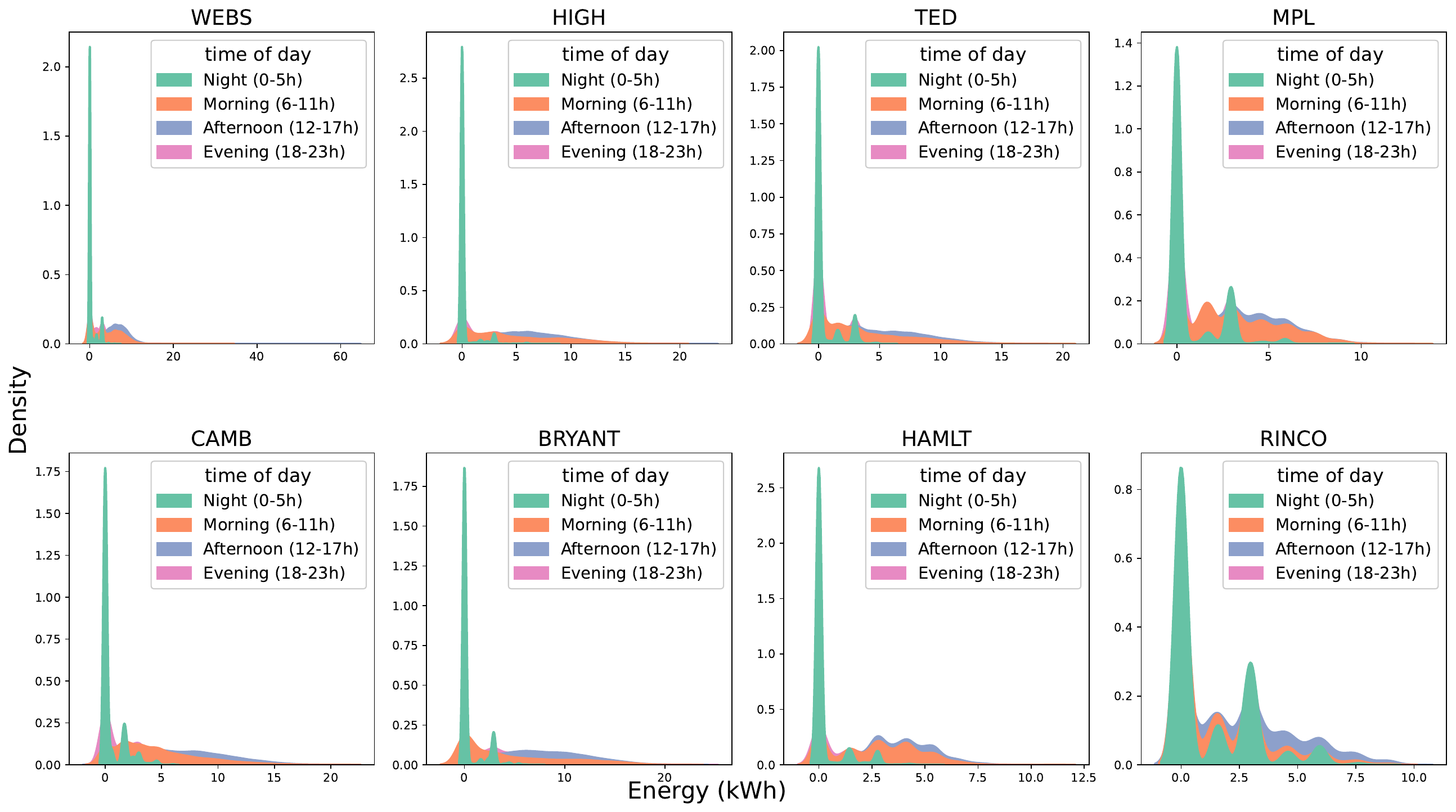}
    \caption{Demand distribution of the stations.}
    \label{fig:demand_distribution}
\end{figure*}

\begin{figure}
    \centering
    \includegraphics[width=6cm,height=5cm]{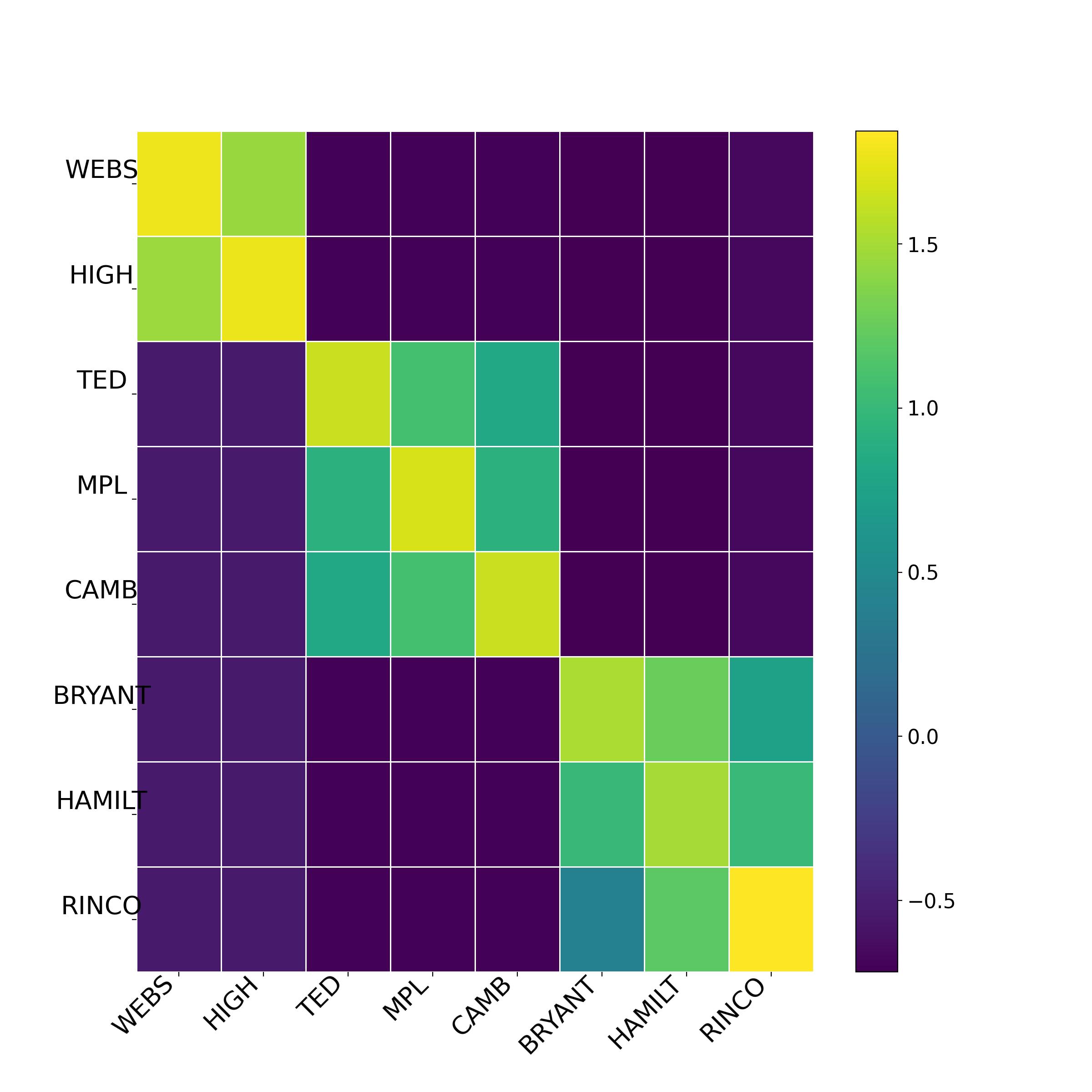}
    \setlength{\abovecaptionskip}{-0.28cm}
    \caption{Heatmap of global attention for Palo Alto dataset}
    \label{fig:heatmap1}
\end{figure}

\subsection {Influence of the Credict Parameter $\chi_i$}
In this paper, $\chi_i$ is an important parameter that client $i$ can use to control the limit of similarity value with other clients. If $\chi_i$ is extremely small, i.e. $\chi_i \rightarrow 0$, which means that client $i$ does not trust any other clients that participate in the FL, and thus it is the same as the local training scenario. On the other hand, if $\chi_i \rightarrow 1$, it means that client $i$ totally trusts at least one client that owns the same data distribution as itself. To investigate the influence of $\chi_i$, for dataset Palo Alto, weight $\xi_{i,j}$ within 100 communication rounds are shown in Fig. \ref{fig:HIGH}. In this case, 'HIGH' station is selected as client $i$, and every 20 communication rounds, one client is designated as malicious, sending a poisoned model to the server. As illustrated in Fig. \ref{fig:HIGH}, three clients possess data distributions similar to 'HIGH' during the training process, under different credit parameters $\chi_i$. For various values of $\chi_i$, client 'HIGH' itself maintains the highest weights $\xi_{i,j}$, which is intuitive as each client trusts itself most. Notably, the weight attributed to itself decreases as $\chi$ increases, indicating greater credit allocation to other clients. When a client transmits malicious information, it is detected by the proposed algorithm, and its weight attributed to 'HIGH' is reset to zero. For example, at the 20th communication round, 'WEBSTER' transmitted a poisoned model to the server for 'HIGH's global model aggregation, resulting in its weight being reduced to nearly zero. We can conclude that the credit-based method demonstrates high effectiveness, consistently excluding malicious clients while assigning different weights based on their contribution/similarity to the selected client and a self-tuned parameter $\chi_i$.

%  When $\chi_i$ is small, e.g. $\chi_i=0.6$, there are only two clients that contribute to 'HIGH'.%indicating 'BRYANT' similar to 'HIGH' most, which accords with the real data distribution. %shown in Fig .\ref{fig:distribution}.
% As the $\chi_i$ increases, more clients are assumed to be beneficial to station 'HIGH', and when $\chi_i=0.95$, all clients except the malicious one are considered as beneficial, with a different but close weight $\xi_{i,j}$. 
% \textcolor{red}{For Shenzhen dataset, weight $\xi_{i,j}$ within 10 communication rounds are shown in Fig. ***. In this case, ***}
% In addition, we have visualized the global attention of the 8 stations as a heatmap, which can be seen in Fig.\ref{fig:heatmap1}. %We also visualize the global attention weights of the Shenzhen dataset as a heatmap (Fig. \ref{fig:heatmap2}). The results show that each node places the highest attention on itself, with significantly lower weights assigned to other nodes. This pattern highlights the strong non-IID characteristics of the dataset.
% \begin{figure}
%     \centering
%     \includegraphics[width=6cm,height=5cm]{figures/FL_ADA_G_predict_all_NEW.jpg}
%     \caption{Heatmap of global attention for Palo Alto dataset}
%     \label{fig:heatmap1}
% \end{figure}

% \begin{figure}
%     \centering
%     \includegraphics[width=6cm,height=5cm]{figures/FL_ADA_G_predict_all_sz1.jpg}
%     \caption{Heatmap of global attention for Shenzhen dataset}
%     \label{fig:heatmap2}
% \end{figure}

% \vspace{-0.8cm}
\begin{figure}
    \centering
    \includegraphics[width=9cm,height=6cm]{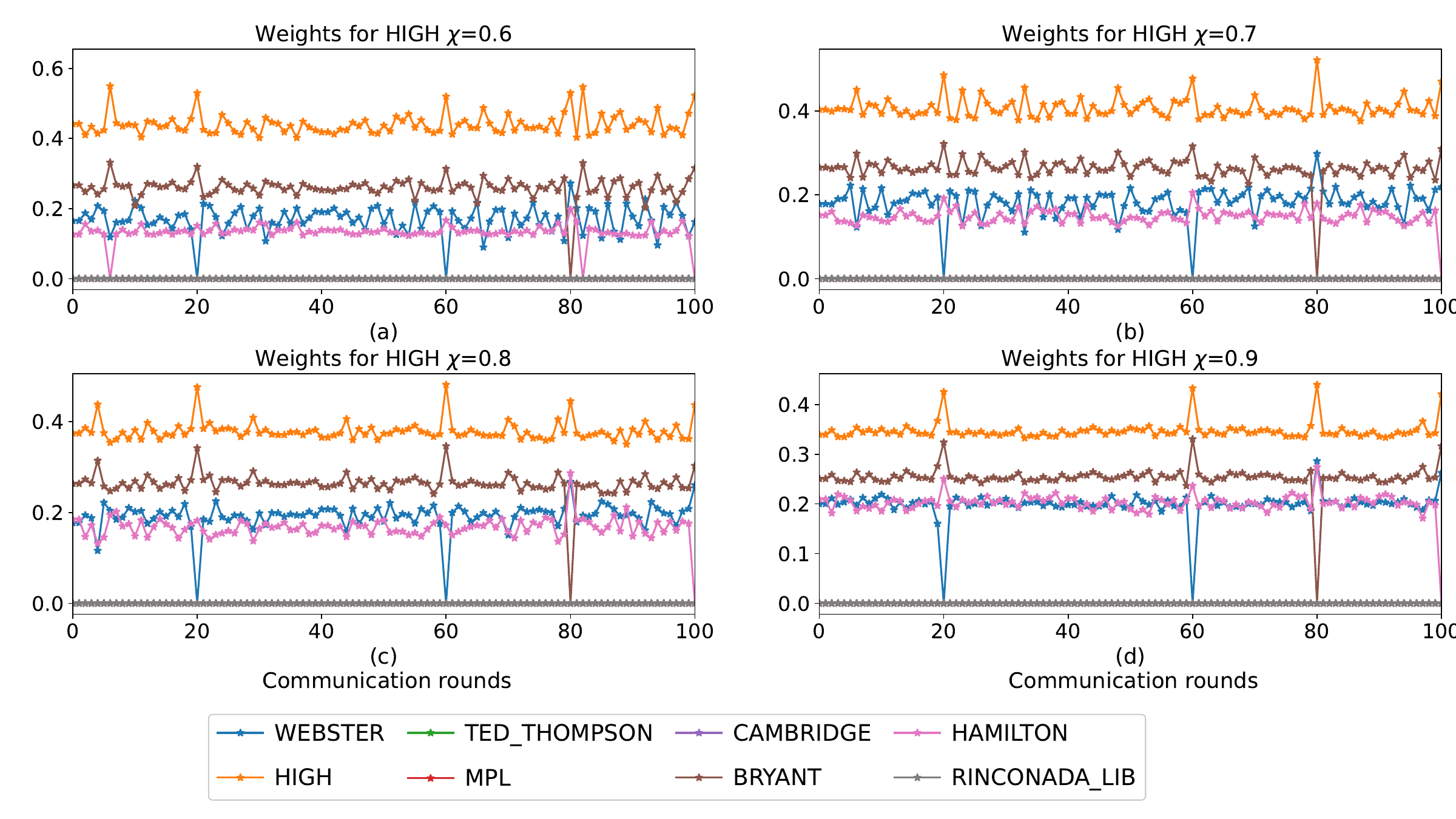}
    \setlength{\abovecaptionskip}{-0.28cm}
    \caption{Weight $\xi_{i,j}$ attributed to 'HIGH' station (a) $\chi=0.6$. (b) $\chi=0.7$, (c) $\chi=0.8$, (d) $\chi=0.9$.}
    \label{fig:HIGH}
\end{figure}

\section{\textcolor{black}{Discussion}}
\label{discussion}
\textcolor{black}{The results show that PFGL improves both forecasting accuracy  under non-IID and adversarial conditions. 
These improvements are not limited to forecasting accuracy, but directly impact the operation of ITS, where charging demand uncertainty and data heterogeneity remain major challenges.}

\textcolor{black}{At the operational level, more accurate and personalized forecasts enable reliable load scheduling, reduce transformer overload risks, and lower operating costs through better reserve planning and dynamic pricing. }

\textcolor{black}{At the planning level, PFGL helps identify long-term spatial–temporal demand patterns for effective charger deployment and grid expansion, while quantified uncertainty bounds reduce over-design and improve coordination with renewable energy resources.}

\textcolor{black}{Despite these advantages, several limitations remain. PFGL relies on the availability of sufficiently representative historical demand and spatial information; abrupt changes in user behavior, policy, or charging technologies may reduce predictive performance. Furthermore, the current uncertainty modeling is based on quantile estimation and does not explicitly capture evolving distributional dynamics over time. These limitations motivate future work on adaptive learning mechanisms and variance-aware or Bayesian extensions to further strengthen robustness and long-term reliability.}

\section{Conclusions}
\label{section5}
In this paper, we propose a personalized federated graph learning (PFGL) framework for multi-horizon EV charging demand forecasting, designed to enhance the operational intelligence of ITS infrastructure while ensuring data privacy and security. By modeling each charging station as a node within a spatiotemporal graph-structured network, our approach captures complex inter-station correlations critical for urban mobility management. The framework integrates Graph Neural Networks (GNNs) into a federated learning paradigm, employing a similarity-based aggregation mechanism and a credit-based adaptive weighting strategy to address the heterogeneous charging patterns inherent in diverse urban locales. Theoretical convergence is established, and extensive experiments on real-world datasets demonstrate that PFGL significantly improves prediction accuracy and provides robust defense against adversarial cyberattacks in connected transportation environments.
\textcolor{black}{Future work will extend PFGL to multi-task and cross-domain smart energy-transportation applications, further refining the management of heterogeneous demand via adaptive personalization and client clustering. Additionally, we plan to integrate variance-aware uncertainty estimation (e.g., Bayesian GNNs or dropout calibration) to provide more reliable interval predictions, thereby supporting more resilient decision-making in large-scale intelligent mobility ecosystems.}
\vspace{-0.6cm}
\bibliographystyle{IEEEtran}
\bibliography{reference}

\newpage
\newpage
\clearpage 
\appendices
%\section{Proof of Theorem \ref{theo:convergence}}
%*********************************
\section{Convergence Analysis}
This subsection provides the convergence analysis of the proposed method.For convenience of convergence analysis, the following assumptions are made.
\begin{assumption}\label{asp:L-smooth}
Function $f_i$ is L-smooth for $i \in [N]$, i.e., there exists a constant $L$, for all  $\bold{x}, \bold{y} \in \mathbb{R}^d$, $f_i(\bold{y}) \le f_i(\bold{x})+(\bold{y-x})^T\nabla f_i(\bold{x})+\frac{L}{2}\|\bold{y-x}\|^2$.
\end{assumption}

\begin{assumption}\label{asp:bounded_grad}
The gradient of $f_i$ is bounded, i.e., $\|\nabla f_i(\bold{w}_i(\tau))\|\le M$ for every $\tau >0$.
\end{assumption}

\begin{assumption}\label{asp:bounded_diff}
The difference between the global model and local model of the client $i$ is bounded for $i \in [N]$, i.e., there exists a constant $Q$ such that $\|\bold{w}_i(\tau)-\bold{w}^G_i(\tau)\|\le \chi_i \eta_{\tau} Q$ for every $t >0$, where $\eta_{\tau}$ is the local learning rate.
\end{assumption}

{\color{black}
\begin{assumption}\label{asp:lipschitz_encoder}
(Lipschitz Continuity of Hidden Representation Encoder) The encoder function that maps raw data to hidden representations is Lipschitz continuous with constant $\rho_h$. Formally, $\forall \textbf{\textit{X}}^{h}_{i}, \textbf{\textit{X}}^{h'}_{i}$, $\exists \rho_h > 0$ such that:
$$\|H_{i}^{h} - H_{i}^{h'}\| \leq \rho_h \|\textbf{\textit{X}}^{h}_{i} - \textbf{\textit{X}}^{h'}_{i}\|$$
\end{assumption}

\begin{assumption}\label{asp:bounded_rep_error}
(Bounded Hidden Representation Estimation Error) When client $j$'s hidden representation is estimated by client $i$ through the aggregation process, the mean squared error is bounded:
$$\mathbb{E}\Big[\|{\hat{H}}^{h}_{j \rightarrow i} - H^{h}_{j}\|^2\Big] \leq \delta_H^2$$
where ${\hat{H}}^{h}_{j \rightarrow i}$ is the estimation of $H^{h}_{j}$ at client $i$.
\end{assumption}

\begin{theorem}\label{theo:convergence}
   Suppose that Assumptions \ref{asp:L-smooth}, \ref{asp:bounded_grad}, \ref{asp:bounded_diff}{\color{black}, \ref{asp:lipschitz_encoder}, and \ref{asp:bounded_rep_error}} hold, if $\eta_{\tau} \le \min\{\frac{1}{LK},\frac{1}{\beta_i \sqrt{3(K-1)(K+1)}}\}$ and if $\eta_{\tau} = \frac{1}{\sqrt{T}}$, then for $T \ge \max\{L^2K^2, 3\beta_i^2 (K-1)(K+1)\}$, Algorithm \ref{alg:server} ensures:
\begin{align}
    &\min_{0 \le t\le T-1} \mathbb{E}[\|\nabla f_i(\bold{w}_i(\tau))\|^2]  \notag \\
    &\le \frac{2 (\mathbb{E}[f_i(\bold{w}_i(0))] - f_i^*)}{K \sqrt{T}} + \frac{\Phi( \beta_i,K) + \Psi( \beta_i,K,\chi_i) {\color{black} + \Gamma(\beta_i,K)}}{T},\notag
\end{align}
where $f_i^*$ is the optimal value of $f_i()$, $\Phi(\beta_i,K)=27(L^2+\beta_i^2)(K+1)M^2$, $\Psi(\beta_i,K,\chi_i)= (\frac{9(L^2+\beta_i^2)}{\beta_i^2(K-1)}+3\beta_i^2)\chi_i^2Q^2${\color{black}, and $\Gamma(\beta_i,K) = 54(L^2+\beta_i^2)(K+1)\rho_g^2\delta_H^2 + \frac{6\rho_g^2\delta_H^2}{K}$ represents the additional error term due to hidden representation estimation}.

Moreover, if $\eta_{\tau}$ satisfies $\sum_{\tau=0}^{\infty}\eta_{\tau} = \infty$, and $\sum_{\tau=0}^{\infty}\eta_{\tau}^3 < \infty$, and $\delta_H \to 0$ as $T \to \infty$, then we have $\lim_{\tau \rightarrow \infty} \mathbb{E}[\|\nabla f_i(\mathbf{w}_i(\tau))\|^2]=0$.
\end{theorem}

\begin{proof}
   % the proof can be found in our arXiv version in \cite{li2024federated}
   %the proof can be found in Appendix B.
%\end{proof}

Based on the L-smooth assumption, 
\begin{align}\label{eqn:L-smooth}
    &f_i(\bold{w}_i(\tau+1)) \notag \\
    \le  &f_i(\bold{w}_i(\tau)) + \langle \nabla f_i(\bold{w}_i(\tau)), \bold{w}_i(\tau+1)-\bold{w}_i(\tau) \rangle \notag \\
    &+ \frac{L}{2} \|\bold{w}_i(\tau+1)-\bold{w}_i(\tau)\|^2 \notag \\
    = &f_i(\bold{w}_i(\tau))+ \langle \nabla f_i(\bold{w}_i(\tau)), -K\eta_{\tau} \bold{d}_t \rangle 
    + \frac{LK^2\eta_{\tau}^2}{2} \|\bold{d}_t\|^2 \notag \\
    {=}&f_i(\bold{w}_i(\tau)) + \frac{LK^2\eta_{\tau}^2-K\eta_{\tau}}{2} \|\bold{d}_t\|^2-\frac{\eta_{\tau} K}{2}\|\nabla f_i(\bold{w}_i(\tau))\|^2\notag \\ 
    &+\frac{\eta_{\tau} K}{2}\underbrace{\|\nabla f_i(\bold{w}_i(\tau))-\bold{d}_t\|^2}_{A}, 
\end{align}
where $\bold{d}_t=[\bold{w}_i(\tau)-\bold{w}_i(\tau+1)]/\eta_{\tau} K$, and the last equality is obtained by using $-\langle a, b \rangle = -\frac{1}{2}(\|a\|^2-\|a-b\|^2+\|b\|^2)$. 

The key difference with hidden representation sharing is that the gradient computation depends on hidden representations, so we need to account for estimation errors. Denoting the gradient based on true hidden representations as $\nabla f_i(H^D_i, \bold{w}_i(\tau))$ and the gradient based on estimated hidden representations as $\nabla \hat{f}_i(\hat{H}^D_i, \bold{w}_i(\tau))$, we have:

\begin{align}
    & \nabla f_i(\bold{w}_i(\tau)) = \nabla f_i(H^D_i, \bold{w}_i(\tau))  \\
    & \bold{d}_t = \frac{1}{K}\sum_{k=0}^{K-1} \Big[\nabla \hat{f}_i(\hat{H}^D_i, \bold{w}_i^k(\tau)) + \beta_i(\bold{w}_i^k(\tau) - \bold{w}_i^G(\tau))\Big]
\end{align}

Now we bound $A$:
\begin{align}\label{eqn:A}
    &\|\nabla f_i(\bold{w}_i(\tau))-\bold{d}_t\|^2 \notag \\
     \overset{a_1}{\le} & \sum_{k=0}^{K-1} \frac{1}{K}\Big(\|\nabla f_i(\bold{w}_i(\tau))-\nabla \hat{f}_i(\hat{H}^D_i, \bold{w}_i^k(\tau))- \notag \\
     & \quad \beta_i(\bold{w}_i^k(\tau)-\bold{w}_i^G(\tau))\|^2\Big)\notag \\
      \overset{a_2}{\le} &\sum_{k=0}^{K-1} \frac{3}{K}\Big(\|\nabla f_i(H^D_i, \bold{w}_i(\tau))-\nabla \hat{f}_i(\hat{H}^D_i, \bold{w}_i^k(\tau))\|^2 \notag \\
      &+ \beta_i^2\|\bold{w}_i^k(\tau)-\bold{w}_i(\tau)\|^2\Big)+ 3\beta_i^2\|\bold{w}_i^G(\tau)-\bold{w}_i(\tau)\|^2 \notag \\
\end{align}

We can further decompose the first term to account for both parameter difference and hidden representation estimation error:

\begin{align}
    &\|\nabla f_i(H^D_i, \bold{w}_i(\tau))-\nabla \hat{f}_i(\hat{H}^D_i, \bold{w}_i^k(\tau))\|^2 \notag \\
    \leq &2\|\nabla f_i(H^D_i, \bold{w}_i(\tau))-\nabla f_i(H^D_i, \bold{w}_i^k(\tau))\|^2 \notag \\ 
    & + 2\|\nabla f_i(H^D_i, \bold{w}_i^k(\tau))-\nabla \hat{f}_i(\hat{H}^D_i, \bold{w}_i^k(\tau))\|^2 \notag \\
    \leq &2L^2\|\bold{w}_i^k(\tau)-\bold{w}_i(\tau)\|^2 + 2\rho_g^2\|H^D_i - \hat{H}^D_i\|^2.
\end{align}
where $\rho_g$ is the Lipschitz constant for the gradient with respect to hidden representations, and we've used Assumption 4.%\ref{asp:lipschitz_encoder}.

%From Assumption \ref{asp:bounded_rep_error}, we have:
From Assumption 5, we have:
\begin{align}
    \mathbb{E}\Big[\|H^D_i - \hat{H}^D_i\|^2\Big] \leq \delta_H^2
\end{align}

Substituting back:
\begin{align}
    A \overset{a_3}{\le}& \frac{3}{K} (2L^2+\beta_i^2) \underbrace{\sum_{k=0}^{K-1} \|\bold{w}_i^k(\tau)-\bold{w}_i(\tau)\|^2}_{B} + \frac{6\rho_g^2\delta_H^2}{K} + 3\eta_{\tau}^2 \beta_i^2 \chi_i^2 Q^2,
\end{align}

Now we will bound B:

\begin{align}
\varepsilon_i^k =&\|\bold{w}_i^k(\tau)-\bold{w}_i(\tau)\|^2\notag \\
=&\|\bold{w}_i^{k-1}(\tau)-\bold{w}_i(\tau) \notag\\
&-\eta_{\tau}\Big(\nabla \hat{f}_i(\hat{H}^D_i, \bold{w}_i^{k-1}(\tau))+\beta_i(\bold{w}_i^{k-1}(\tau)-\bold{w}_i^G(\tau))\Big)\|^2 \notag \\
\overset{b_1}{\le} & (1+\frac{1}{K}) \varepsilon_i^{k-1} \notag \\
&+ (K+1)\eta_{\tau}^2\|\nabla \hat{f}_i(\hat{H}^D_i, \bold{w}_i^{k-1}(\tau))+\beta_i(\bold{w}_i^{k-1}(\tau)-\bold{w}_i^G(\tau))\|^2 \notag\\
\end{align}

Taking expectation on both sides and accounting for hidden representation estimation error:

\begin{align}
\mathbb{E}[\varepsilon_i^k] \overset{b_2}{\le} & (1+\frac{1}{K}) \mathbb{E}[\varepsilon_i^{k-1}] + 3(K+1)\eta_{\tau}^2\mathbb{E}[\|\nabla \hat{f}_i(\hat{H}^D_i, \bold{w}_i^{k-1}(\tau))\|^2] \notag \\
& + 3(K+1)\eta_{\tau}^2\beta_i^2\Big(\mathbb{E}[\|\bold{w}_i^{k-1}(\tau)-\bold{w}_i(\tau)\|^2] \notag \\
 &+ \|\bold{w}_i(\tau)-\bold{w}_i^G(\tau)\|^2\Big) 
\end{align}

We can bound $\mathbb{E}[\|\nabla \hat{f}_i(\hat{H}^D_i, \bold{w}_i^{k-1}(\tau))\|^2]$ as:

\begin{align}
\mathbb{E}[\|\nabla \hat{f}_i(\hat{H}^D_i, \bold{w}_i^{k-1}(\tau))\|^2] \leq 2M^2 + 2\rho_g^2\delta_H^2.
\end{align}

Substituting this back:

\begin{align}
\mathbb{E}[\varepsilon_i^k] \overset{b_3}{\le} &(\frac{K+1}{K}+3(K+1)\eta_{\tau}^2 \beta_i^2) \mathbb{E}[\varepsilon_i^{k-1}] \notag \\
&+ 3(K+1)\eta_{\tau}^2(2M^2 + 2\rho_g^2\delta_H^2 + \beta_i^2\eta_{\tau}^2 \chi_i^2Q^2) \notag \\
\overset{b_4}{\le} & (1+\frac{1}{K-1})^{K-1} \times 3(K+1)\eta_{\tau}^2(2M^2 + 2\rho_g^2\delta_H^2 + \eta_{\tau}^2 \chi_i^2Q^2) \notag \\
\overset{b_5}{\le} &9(K+1)\eta_{\tau}^2(2M^2 + 2\rho_g^2\delta_H^2 + \eta_{\tau}^2 \chi_i^2Q^2),
\end{align}
where we've used $\eta_{\tau} \beta_i \leq \frac{1}{\sqrt{3(K-1)(K+1)}}$ and $(1+\frac{1}{x})^x < 3$ for $x > 0$.

Therefore:
\begin{align}\label{eqn:B}
    B=\sum_{k=0}^{K-1}\mathbb{E}[\epsilon_i^k] \leq 9K(K+1)\eta_{\tau}^2(2M^2 + 2\rho_g^2\delta_H^2 + \eta_{\tau}^2 \chi_i^2Q^2).
\end{align}

Substituting back to bound $A$:
\begin{align}
    \mathbb{E}[A] \leq & 54(L^2+\beta_i^2)(K+1)M^2\eta_{\tau}^2 + 54(L^2+\beta_i^2)(K+1)\rho_g^2\delta_H^2\eta_{\tau}^2 \notag \\
    & + \left(\frac{18(L^2+\beta_i^2)}{\beta_i^2(K-1)}+3\beta_i^2\right)\eta_{\tau}^2 \chi_i^2Q^2 + \frac{6\rho_g^2\delta_H^2}{K}
\end{align}

Let's define:
\begin{align}
\Phi(\beta_i,K) &= 54(L^2+\beta_i^2)(K+1)M^2 \\
\Psi(\beta_i,K,\chi_i) &= \left(\frac{18(L^2+\beta_i^2)}{\beta_i^2(K-1)}+3\beta_i^2\right)\chi_i^2Q^2 \\
\Gamma(\beta_i,K) &= 54(L^2+\beta_i^2)(K+1)\rho_g^2\delta_H^2 + \frac{6\rho_g^2\delta_H^2}{K}
\end{align}

Substituting into the original inequality and taking expectation, we get:
\begin{align}\label{eqn:f-1}
    & \mathbb{E}[f_i(\bold{w}_i(\tau+1))] \leq \mathbb{E}[f_i(\bold{w}_i(\tau))]  -\frac{\eta_{\tau} K}{2}\mathbb{E}[\|\nabla f_i(\bold{w}_i(\tau))\|^2] \notag \\
    &+\frac{\eta_{\tau}^3 K}{2}(\Phi(\beta_i,K) + \Psi(\beta_i,K,\chi_i)) + \frac{\eta_{\tau} K}{2}\Gamma(\beta_i,K)\eta_{\tau}^2
\end{align}

Rearranging and summing over $\tau = 0,1,...,T-1$:
\begin{align}
   & \min_{0 \leq \tau \leq T-1} \mathbb{E}[\|\nabla f_i(\bold{w}_i(\tau))\|^2]  
   \leq \frac{2 (\mathbb{E}[f_i(\bold{w}_i(0))] - f_i^*)}{K \sum_{\tau=0}^{T-1}\eta_{\tau}} \notag \\
   & + \frac{\sum_{\tau=0}^{T-1}{\eta_{\tau}^3}}{\sum_{\tau=0}^{T-1}\eta_{\tau}}(\Phi(\beta_i,K) + \Psi(\beta_i,K,\chi_i) + \Gamma(\beta_i,K)),\notag
\end{align}

where $f_i^*$ is the optimal value of $f_i()$.

If we set $\eta_{\tau} = \frac{1}{\sqrt{T}}$, then:
\begin{align}
    & \min_{0 \leq \tau \leq T-1} \mathbb{E}[\|\nabla f_i(\bold{w}_i(\tau))\|^2]  
    \leq \frac{2 (\mathbb{E}[f_i(\bold{w}_i(0))] - f_i^*)}{K \sqrt{T}} \notag \\
    & + \frac{\Phi(\beta_i,K) + \Psi(\beta_i,K,\chi_i) + \Gamma(\beta_i,K)}{T}. \notag
\end{align}

Moreover, if $\eta_{\tau}$ satisfies $\sum_{\tau=0}^{\infty}\eta_{\tau} = \infty$, and $\sum_{\tau=0}^{\infty}\eta_{\tau}^3 < \infty$, and $\delta_H \to 0$ as $T \to \infty$, then we have $\lim_{\tau \rightarrow \infty} \mathbb{E}[\|\nabla f_i(\mathbf{w}_i(\tau))\|^2]=0$.
\end{proof}

\section{Experiments}
\label{appendixb}
\subsection{Experimental Setups}
\textbf{Data Preprocessing:}
    \textcolor{black}{The raw charging records are processed as follows.
Missing values are filled by forward interpolation, and abnormal values beyond three standard deviations are clipped. Incomplete daily records are removed. Each station’s demand series is normalized using z-score.
To capture periodic patterns in charging behavior, calendar features are extracted and encoded as additional inputs, including hour of day, day of week, day of month, day of year, month, and year. 
After feature construction, the time series data are segmented using a sliding window mechanism. 
Specifically, a historical window of length $h$ is used to predict future demand, where $h=12$ for the step-1 forecasting task and $h=32$ for the step-6 forecasting task. 
The resulting samples are organized into mini-batches of size 32 and fed into the model using a data loader during training and evaluation.}

\textbf{Data Distribution and Heterogeneity Analysis:}
\textcolor{black}{To quantify spatial heterogeneity of EV charging demand across stations, we compute the average Jensen–Shannon (JS) distance between the empirical demand distributions of different stations within each dataset.
The results indicate clear distributional differences: Palo Alto exhibits the highest heterogeneity (average JS = 0.4206), followed by UrbanEV (JS = 0.2945), while Shenzhen is considerably more homogeneous (JS = 0.2585). The standard deviations of JS further confirm this pattern, with Shenzhen showing relatively concentrated distributional similarity among stations.
These statistics demonstrate that Shenzhen represents a more homogeneous urban charging environment, whereas UrbanEV and Palo Alto present substantially stronger spatial heterogeneity.}

\textbf{Implementation Details:}
    \textcolor{black}{Experiments were conducted for both single-step (H=1) and multi-step (H=6) forecasting horizons. We set batch size as 32 across all models. Learning rates were tuned based on dataset characteristics. For Palo Alto dataset, No\_FL, CNFGNN, PAG, GCRN, \textcolor{black}{FedProx, pFedMe} and PFGL utilized a rate of 0.0005, while LSTM-FedAvg, LSTM-CNN-FedAvg, and LSTM-CNN-PFL used 0.001. For Shenzhen and \textcolor{black}{UrbanEV} dataset, CNFGNN, PAG, GCRN, and PFGL employed 0.005, with the LSTM-based models again set to 0.001, \textcolor{black}{FedProx and pFedMe used 0.0005}. \textcolor{black}{For FedProx, the proximal coefficient is set to $0.001$. For pFedMe, we set personalization strength $0.001$ for the highly non-IID charging demand setting, inner-loop learning rate $0.0001$.}
    Input sequences were generated via a sliding window approach. The window size of 12 time steps was chosen for step-1 forecasting to capture immediate temporal dynamics, whereas a larger window of 32 steps was used for step-6 to incorporate longer-term dependencies potentially crucial for extended predictions.} 
For No\_FL, LSTM-CNN-FedAvg and LSTM-CNN-PFL, the kernel size of 1D Conv modules are both 64 and the Linear networks have 64 hidden units. There exists 2 layer LSTM module in LSTM-FedAvg, LSTM-CNN-FedAvg and LSTM-CNN-PFL, the hidden size of each layer are 32 and 64, respectively. \textcolor{black}{For other GNN-based federate learning methods,
they employs 1D Conv Encoder and Decoder have a kernel size of 64 (for step-1) and 128 (for step-6) to capture long temporal dependencies.  For the hyperparameters $\alpha$ and $\chi_i$, we adopted the following configurations: for Palo Alto dataset, in single-step prediction, $\alpha$=0.8, $\chi_i$=0.8.
%In 3-step prediction: $\alpha$=0.9, $\sigma_{i}$=0.8. 
In 6-step prediction, $\alpha$=0.9, $\chi_i$=0.9.
For Shenzhen and \textcolor{black}{UrbanEV} dataset, $\alpha$ = 0.9, $\chi_i$=0.4.} 
In addition, we set $\theta$ to  $\{0.1, 0.5, 0.9\}$ to comprehensively assess the model's capability in providing probabilistic predictions across different levels of confidence intervals. 
 To ensure model generalization and mitigate overfitting, we divided the dataset into a training set (60\%), a validation set (20\%), and a testing set (20\%).
We perform 100 communication rounds and select the model with minimal global validation loss as the final model. 

\subsection{\textcolor{black}{Influence of the Graph neural networks block}}
\textcolor{black}{Several existing works utilize graph-based methods for demand prediction tasks. To evaluate the effectiveness of employing GNN for capturing latent spatial-temporal features within our proposed personalized federated learning (PFL) framework, we conducted a comparative study. We integrated different state-of-the-art  spatial-temporal graph neural network architectures (CNFGNN, PAG, and GCRN) into a common PFL base framework (denoted CNFGNN-PFL, PAG-PFL, GCRN-PFL). We then compared their performance against our proposed model, PFGL. All experiments were performed without adversarial attacks. According to the result (Table. X), the proposed model achieves the best results on QS and MIL across both datasets, and also achieves the best ICP on the Palo Alto dataset \textcolor{black}{and UrbanEV dataset}. 
These results verify the effectiveness of the graph learning module embedded within PFGL. Despite leveraging similar underlying GNN structures, the spatial-temporal attention message passing in PFGL enables client to explicitly gather information from both current spatial neighbors and historical states via a temporal attention mechanism.  
This design allows it to better capture and exploit latent spatial-temporal dependencies compared to existing GNN-based PFL baselines. Moreover, the GraphSAGE block used in PFGL benefits from inductive representation learning, enabling generalization to unseen clients.}
%The comparative advantage of PFGL highlights its superior representational capacity and alignment with client-specific characteristics, making it particularly suitable for federated demand forecasting tasks in heterogeneous environments.}
\begin{table}[t]
\setlength{\abovecaptionskip}{0cm}
\setlength{\belowcaptionskip}{-0.6cm}
\centering
\vspace{-0.4cm}
\color{black}{
\caption{QS, MIL and ICP with different GNN-block on different dataset (Prediction Step=6)}
\scalebox{0.8}{
\scriptsize
\begin{tabular}{c c|ccc}
\hline \hline
\textbf{Dataset} & \textbf{Model} 
& \textbf{QS} & \textbf{MIL} & \textbf{ICP} \\
\hline
\multirow{3}{*}{Palo Alto} 
& CNFGNN-PFL & 1.3285 & 4.1295 & 0.7651 \\
& PAG-PFL & 1.6089 & 4.929 & 0.7398 \\
& GCRN-PFL & 1.4238 & 4.2057 & 0.7395 \\
& PFGL & \textbf{1.318} & \textbf{3.9813} & \textbf{0.7684} \\
\hline
\multirow{3}{*}{Shenzhen} 
& CNFGNN-PFL & 10.3131 & 9.4765 & 0.7018 \\
& PAG-PFL & 15.1751 & 25.733 & \textbf{0.7078} \\
& GCRN-PFL & 9.8595 & 9.4765 & 0.6927 \\
& PFGL & \textbf{9.4411} & \textbf{8.5866}& 0.7002\\
\hline 
\multirow{3}{*}{Urban} 
& CNFGNN-PFL & 3.4075 & 33.6296 & 0.7843 \\
& PAG-PFL & 3.469 & 32.5507 & 0.7327 \\
& GCRN-PFL & 3.3513 & 32.4954 & 0.7919 \\
& PFGL & \textbf{2.3077} & \textbf{23.5551}& \textbf{0.7944}\\
\hline\hline
\end{tabular}
}
}
\label{GNN_block_X}
\end{table}
\subsection{\textcolor{black}{Influence and Strategy for GNN over-smoothing Problem}}

\color{black}{GNN over-smoothing, characterized by the convergence of node representations towards similar values after multiple layers of message passing, presents a challenge. To mitigate the risk of over-smoothing,  we set GNN-layer as 3 and residual structure is employed in our model architecture. Specifically, the initial client representations are combined with the final output of the GNN layers through a residual connection, and the resulting features are then passed to the prediction stage. This mechanism allows initial client feature information to
propagate directly, preserving local specificity and counteracting the homogenizing effect of iterative
message passing.} 

\color{black}{To verify the effectiveness of these strategy, we have conduct some ablation experiments on Shenzhen dataset. Different GNN layer with/without residual structure. As shown in Table~\ref{residual}, increasing the number of layers to 4 or 5 without residual connections leads to significant performance degradation. For example, with 5 layers and no residual structure, the PFGL method yields a QS of 19.1585, MIL of 42.844. In contrast, our proposed residual-enhanced architecture maintains strong performance at 3 layers, achieving a good balance between expressiveness and stability.  These results confirm that residual connections are crucial for maintaining model effectiveness, particularly when exploring deeper GNNs, and validate our architectural choice of a 3-layer GNN with residuals in PFGL to prevent over-smoothing while preserving rich spatial-temporal representations.}
\begin{table}[]
\setlength{\abovecaptionskip}{0cm}
\centering
	\caption{Influence of the residual structure for GNN oversmooth (Prediction Step = 6)}
        \label{residual}
        \scalebox{0.8}{
	\begin{tabular}{l l c c c}
		\hline \hline
		         Layer num  & Method  & QS & MIL & ICP \\ 
		\hline
		\multirow{2}{*}{3} 
            & w/o residual & 9.9072    & 8.8304  & \textbf{0.7445}  \\          
		& PFGL &\textbf{9.4411} &\textbf{8.5866} &0.7002 \\	
        \hline
        \multirow{2}{*}{4} 
            & w/o residual & 17.4247    & 29.418  & \textbf{0.7153}  \\          
		& PFGL &\textbf{9.8542} &\textbf{9.0629} &0.6946 \\
        \hline
        \multirow{2}{*}{5} 
            & w/o residual  &19.1585   &42.844  &0.7312\\
            & PFGL &\textbf{13.3638} &\textbf{18.6912} &\textbf{0.7046} \\
		\hline \hline
	\end{tabular}}
\end{table}

\subsection{\textcolor{black}{Influence of the channel weight $\alpha$}}
\textcolor{black}{In this paper, parameter $\alpha$, termed the 'channel weight' in Eq. (\eqref{channel_weight}), plays a crucial role in balancing the influence of two distinct information sources when determining the personalized aggregation weights ($\lambda_{i}$) for each client $i$. Specifically, it balances the similarity weighted adjacency matrix based on model similarity ($\xi$) and the spatial weighted matrix based on spatial similarity ($\zeta$). When $\alpha$ approaches 1, the aggregation weight $\xi_i$ is dominated by original $\xi_i$, meaning that clients with more similar models contribute more significantly to each other’s updates, irrespective of their physical location.
Conversely, when $\alpha$ approaches 0, the aggregation weight $\xi_i$ relies primarily on $\zeta$, giving greater influence to clients that are spatially closer, regardless of their current model similarity.
To investigate the optimal balance between these two factors, experiments were conducted on the Shenzhen dataset, varying $\alpha$ across the set $\{0, 0.2, 0.5, 0.9, 1\}$. The performance, measured by QS, MIL, and ICP along with the percentage improvement ('Improv.') over the FedAvg-GraphSAGE baseline, is presented in Table \ref{alpha}. 
Therefore, we choose the $\alpha=0.9$, which can balance the spatial relationship and model similarity.}
\begin{table}[htbp]
\centering
\caption{Performance under different channel weights for Shenzhen dataset (Prediction Step = 6)}
\scalebox{0.8}{
\begin{tabular}{c|cc|cc|cc}
\hline \hline
\textbf{Method} & \multicolumn{2}{c|}{\textbf{QS}} & \multicolumn{2}{c|}{\textbf{MIL}} & \multicolumn{2}{c}{\textbf{ICP}} \\
 & Value & Improv. & Value & Improv. & Value & Improv. \\
 \hline
channel\_weight = 0   & 9.5661 & 51.5\% & 8.9963 & 51\% & 0.7176 & 51.5\% \\
channel\_weight = 0.2 & 9.5774 & 79\% & 8.6931 & 52\% & 0.7164 & 51\% \\
channel\_weight = 0.5 & 9.5798 & 75\% & 8.9568 & 47.5\%  & \textbf{0.7202} & 53.5\% \\
channel\_weight = 0.9 & \textbf{9.4411} & \textbf{84\%} & \textbf{8.5866} & \textbf{63\% }& 0.7002 & \textbf{56.5\%} \\
channel\_weight = 1   & 9.5057 & \textbf{84\%} & 8.7907 & 59.5\% & 0.7001 & \textbf{56.5\%} \\
\hline \hline
\end{tabular}}
\label{alpha}
\end{table}
\subsection{Influence of the increasing number of EV charging stations}
To evaluate the scalability of our approach, we increased the number of participating charging stations in both the Palo Alto (Table \ref{number}),  Shenzhen (Table \ref{number2}) and \textcolor{black}{UrbanEV (Table \ref{number3}) datasets}. 
In this ablation study, all stations in the test set were also seen during training, no unseen stations were introduced to focus on scalability. %As a result, the results with 200 stations are better than the previous results without attacks \ref{QS_2}.  

The results consistently demonstrate that forecasting accuracy improves as more clients participate. Specifically, for the Palo Alto dataset, increasing the number of stations leads to lower QS loss and an ICP value that approaches 0.8. A similar trend is observed in the Shenzhen \textcolor{black}{and UrbanEV dataset}, where prediction errors steadily decline with the inclusion of more clients. These consistent improvements across multiple evaluation metrics (QS, MIL) suggest that the federated model effectively exploits the broader data distribution from a larger client pool, resulting in a more accurate and generalized global model.

\begin{table}[]
\setlength{\abovecaptionskip}{0cm}
\setlength{\belowcaptionskip}{-0.6cm}
\centering
\vspace{-0.4cm}
	\caption{QS, MIL, ICP for increasing number of EV charging stations for Palo Alto dataset (Prediction Step=6)}
        \label{number}
        \scalebox{0.8}{
	\begin{tabular}{@{}c| c c c| c c c@{}}
		\hline \hline
		\multirow{2}{*}{Num of Stations} &  & WEBSTER &  &  &MPL & \\ 
            \multirow{1}{*}{}
                &QS &MIL &ICP & QS &MIL &ICP \\
            \hline
		\multirow{2}{*}{}
            2 & 1.4733 &4.2793 &0.7088 & 1.1407 & \textbf{3.5261} &0.7264 \\
	      8 & \textbf{1.3768}  &\textbf{4.1713} &\textbf{0.7749} &\textbf{1.1325} &3.5485 &\textbf{0.7783} \\
		\hline \hline
	\end{tabular}}
\end{table}

\begin{table}[htbp]
\setlength{\abovecaptionskip}{0cm}
\setlength{\belowcaptionskip}{-0.6cm}
\centering
\vspace{-0.4cm}
\caption{QS, MPL, and ICP for increasing number of EV charging stations on Shenzhen dataset (Prediction Step = 6)}
\label{number2}
\scalebox{0.8}{
\begin{tabular}{@{}c|cc|cc|cc@{}}
\hline \hline
\multirow{2}{*}{Num of Stations} & \multicolumn{2}{c|}{QS} & \multicolumn{2}{c|}{MPL} & \multicolumn{2}{c}{ICP} \\
 & Mean & Std & Mean & Std & Mean & Std \\
\hline
50 & 10.2764 & \textbf{22.8216} & 22.3021 & 42.8743 & 0.7648 & \textbf{0.1106} \\
100 & 12.2052 & 37.9423 & 17.3024 & 39.7369 & 0.7518 & 0.1624 \\
200 & \textbf{7.8118} & 27.5722 &\textbf{12.6685}& \textbf{31.5399} & \textbf{0.7401 }& 0.1471 \\
\hline \hline
\end{tabular}}
\end{table}

\begin{table}[htbp]
\setlength{\abovecaptionskip}{0cm}
\setlength{\belowcaptionskip}{-0.6cm}
\centering
\color{black}{
\vspace{-0.4cm}
\caption{QS, MPL, and ICP for increasing number of EV charging stations on UrbanEV dataset (Prediction Step = 6)}
\label{number3}
\scalebox{0.8}{
\begin{tabular}{@{}c|cc|cc|cc@{}}
\hline \hline
\multirow{2}{*}{Num of Stations} & \multicolumn{2}{c|}{QS} & \multicolumn{2}{c|}{MPL} & \multicolumn{2}{c}{ICP} \\
 & Mean & Std & Mean & Std & Mean & Std \\
\hline
125 & 2.3184 & 7.7567 & \textbf{21.696} & \textbf{68.7079} & 0.8357 & \textbf{0.1056} \\
250 & \textbf{2.3077} & \textbf{7.14} & 23.5551 & 72.4787 & \textbf{0.7944} & 0.1077 \\
\hline \hline
\end{tabular}}}
\end{table}
\appendices

\end{document}